\documentclass[12pt]{article}
\usepackage{amsxtra}
\usepackage{amssymb, amsmath}
\usepackage{amsthm}
\usepackage {hyperref}

\newtheorem{theorem}{Theorem}[section]
\newtheorem{remark}{Remark}[section]

\usepackage[dvips]{graphicx}
\usepackage {multicol}

\begin{document}
\begin{center}
\textbf{\LARGE{Refinement of Gini-Means Inequalities and Connections with Divergence Measures}}
\end{center}

\bigskip
\begin{center}
\textbf{\large{Inder Jeet Taneja}}\\
Departamento de Matem\'{a}tica\\
Universidade Federal de Santa Catarina\\
88.040-900 Florian\'{o}polis, SC, Brazil.\\
\textit{e-mail: ijtaneja@gmail.com\\
http://www.mtm.ufsc.br/$\sim $taneja}
\end{center}

\begin{abstract}
In 1938, Gini \cite{gin} studied a mean having two parameters. Later, many authors studied properties of this mean. It contains as particular cases the famous means such as \textit{harmonic, geometric, arithmetic}, etc. Also it contains, the \textit{power mean of order r }and \textit{Lehmer mean} as particular cases. In this paper we have considered inequalities arising due to Gini-Mean and Heron's mean, and improved them based on the results recently studied by the author \cite{tan5}.
\end{abstract}



\section{Gini Mean of order \textit{r} and \textit{s}}

The Gini \cite{gin} mean of order $r$ and $s$ is given by
\begin{equation}
\label{eq1}
E_{r,s} (a,b) =\begin{cases}
 {\left( {\frac{a^r + b^r}{a^s + b^s}} \right)^{\frac{1}{r - s}}} & {r \ne
s} \\
 {\exp \left( {\frac{a^r\ln a + b^r\ln b}{a^r + b^r}} \right)} & {r = s \ne
0} \\
 {\sqrt {ab} } & {r = s = 0} \\
\end{cases}.
\end{equation}

In particular when $s = 0$ in (\ref{eq1}), we have
\begin{equation}
\label{eq2}
E_{r,0} (a,b): = B_r (a,b) =\begin{cases}
 {\left( {\frac{a^r + b^r}{2}} \right)^{\frac{1}{r}},} & {r \ne 0} \\
 {\sqrt {ab} ,} & {r = 0} \\
\end{cases}
\end{equation}

Again, when $s = r - 1$ in (\ref{eq1}), we have
\begin{equation}
\label{eq3}
E_{r,r - 1} (a,b): = K_s (a,b) = \frac{a^r + b^r}{a^{r - 1} + b^{r - 1}},\;r
\in {\rm R}
\end{equation}

The expression (\ref{eq2}) is famous \textit{as mean of order r} or \textit{power mean}. The expression (\ref{eq3}) is known as \textit{Lehmer mean}\cite{leh}. Both these means are monotonically increasing in $r$. Moreover, these two have the following inequality \cite{che} among each other:
\begin{equation}
\label{eq4}
B_r (a,b)\begin{cases}
 { < L_r (a,b),} & {r > 1} \\
 { > L_r (a,b),} & {r < 1} \\
\end{cases}
\end{equation}

Since $E_{r,s} = E_{s,r} $, the Gini-mean $E_{r,s} (a,b)$ given by (\ref{eq1}) is an increasing function in $r$ or $s$. Using the monotonicity property \cite{czp}, \cite{sim1}, \cite{san} we have the following inequalities:
\begin{align}
& E_{ - 3, - 2} \le E_{ - 2, - 1} \le E_{ - 3 / 2, - 1 / 2} \le E_{ - 1,0} \le
E_{ - 1 / 2,0} \le\notag\\
\label{eq5}
& \hspace{20pt} E_{ - 1 / 2,1 / 2} \le E_{0,1 / 2} \le E_{0,1} \le \left\{ {E_{0,2} \mbox{
or }E_{1 / 2,1} } \right\} \le E_{1,2}.
\end{align}

Let us write the expression (\ref{eq5}) as
\begin{equation}
\label{eq6}
P_1 \le P_2 \le P_3 \le H \le P_4 \le G \le N_1 \le A \le \left( {P_5 \mbox{
or }S} \right) \le P_6 ,
\end{equation}

\noindent
where $P_1 = E_{ - 3, - 2} = K_{ - 2} $, $P_2 = E_{ - 2, - 1} = K_{ - 1}$,
$P_3 = E_{ - 3 / 2, - 1 / 2} = K_{ - 1 / 2} $, $H = E_{ - 1,0} = K_0 = B_{ -
1} $, $P_4 = E_{ - 1 / 2,0} = B_{ - 1 / 2} $, $G = E_{ - 1 / 2,1 / 2} = K_{1
/ 2} = B_0 $, $N_1 = E_{0,1 / 2} = B_{1 / 2} $, $A = E_{0,1} = K_1 = B_1 $,
$P_5 = E_{1 / 2,1} $, $B_2 = E_{0,2} = S$ and $P_6 = E_{1,2} = K_2 $.

\bigskip
The means $H$, $G$, $A$ and $S$ are the \textit{harmonic}, \textit{geometric}, \textit{arithmetic} and the \textit{square-root} means respectively. In \cite{tan2, tan4}, the author studied the following inequalities:
\begin{equation}
\label{eq7}
H \le G \le N_1 \le N_3 \le N_2 \le A \le S,
\end{equation}

\noindent where
\[
N_2 (a,b) = \left( {\frac{\sqrt a + \sqrt b }{2}} \right)\left( {\sqrt
{\frac{a + b}{2}} } \right)
\]
\noindent and
\[
N_3 (a,b) = \frac{a + \sqrt {ab} + b}{3}.
\]

The expression $N_3 (a,b)$ is famous as Heron's mean. Some applications of the inequalities (\ref{eq7}) can be seen in \cite{sim2}, \cite{szd}. Combining (\ref{eq6}) and (\ref{eq7}), we have the following sequence of inequalities:
\begin{equation}
\label{eq8}
P_1 \le P_2 \le P_3 \le H \le P_4 \le G \le N_1 \le N_3 \le N_2 \le A \le
\left\{ {P_5 \mbox{ or }S} \right\} \le P_6 .
\end{equation}

The expression (\ref{eq8}) admits many non-negative differences. Let us write them
as follows:
\begin{equation}
\label{eq9}
D_{tp} (a,b) = bg_{tp} \left( {\frac{a}{b}} \right) = b\left[ {f_t \left(
{\frac{a}{b}} \right) - f_p \left( {\frac{a}{b}} \right)} \right],
\end{equation}

\noindent where
\[
g_{tp} (x) = f_t (x) - f_p (x),
\,
f_t (x) \ge f_p (x),
\,
\forall x > 0.
\]

More precisely, the function $f:(0,\infty ) \to {\rm R}$ appearing in (\ref{eq9})
lead us to the following inequalities:
\begin{align}
& f_{P_1 } (x) \le f_{P_2 } (x) \le f_{P_3 } (x) \le f_H (x) \le f_{P_4 } (x)
\le f_G (x) \le f_{N_1 } (x) \le \notag\\
& \hspace{20pt} \le f_{N_3 } (x) \le f_{N_2 } (x) \le f_A (x) \le \left\{ {f_{P_5 }
(x)\mbox{ or }f_S (x)} \right\} \le f_{P_6 } (x).\notag
\end{align}

Equivalently, we have
\begin{align}
&\frac{x(x^2 + 1)}{x^3 + 1} \le \frac{x(x + 1)}{x^2 + 1} \le \frac{x\left(
{\sqrt x + 1} \right)}{x^{3 / 2} + 1} \le \frac{2x}{1 + x} \le
\frac{4x}{\left( {\sqrt x + 1} \right)^2} \le \notag\\
& \hspace{20pt} \le \sqrt x \le \left( {\frac{\sqrt x + 1}{2}} \right)^2 \le \frac{x +
\sqrt x + 1}{3} \le \left( {\frac{\sqrt x + 1}{2}} \right)\left( {\sqrt
{\frac{x + 1}{2}} } \right) \le\notag\\
\label{eq10}
& \hspace{30pt} \le \frac{x + 1}{2} \le \left\{ {\left( {\frac{x + 1}{\sqrt x + 1}}
\right)^2\mbox{ or }\sqrt {\frac{x^2 + 1}{2}} } \right\} \le \frac{x^2 +
1}{x + 1}.
\end{align}

Based on the differences arising due to inequalities (\ref{eq8}) written according
to (\ref{eq9}), with the property that the functions are convex, the author
\cite{tan5} proved the following sequences of inequalities:
\begin{align}
& \frac{1}{8}D_{P_6 P_1 } \le \frac{1}{6}D_{P_6 P_2 } \le D_{SA} \le
\frac{1}{3}D_{SH} \le \frac{1}{2}D_{AH} \le\notag\\
& \hspace{15pt} \le \left\{ {\begin{array}{l}
 \textstyle{4 \over 9}D_{P_6 N_2 } \\
 \left\{ {\begin{array}{l}
 \textstyle{3 \over 7}D_{P_6 N_3 } \\
 \textstyle{2 \over 5}D_{SP_4 } \\
 \end{array}} \right\} \le \left\{ {\begin{array}{l}
 \textstyle{2 \over 5}D_{P_6 N_1 } \\
 \textstyle{2 \over 7}D_{P_6 P_4 } \\
 \end{array}} \right. \\
 \end{array}} \right\} \le \frac{1}{3}D_{P_6 G} \le \left\{
{\begin{array}{l}
 \textstyle{2 \over 5}D_{P_5 H} \\
 \textstyle{2 \over 3}D_{AP_4 } \\
 \end{array}} \right\} \le 4D_{N_2 N_1 } \le \notag\\
& \hspace{30pt} \le \frac{4}{3}D_{N_2 G} \le D_{AG} \le 4D_{AN_2 } \le \frac{2}{3}D_{P_5 G}
\le D_{P_5 N_1 } \le \notag\\
\label{eq11}
& \hspace{45pt} \le \frac{6}{5}D_{P_5 N_3 } \le \frac{4}{3}D_{P_5 N_2 } \le 2D_{P_5 A},\\\notag\\
\label{eq12}
& D_{SA} \le \left\{ {\begin{array}{l}
 \textstyle{4 \over 5}D_{SN_2 } \\
 \textstyle{3 \over 4}D_{SN_3 } \\
 \end{array}} \right\} \le \frac{2}{3}D_{SN_1 } \le \left\{
{\begin{array}{l}
 \textstyle{1 \over 3}D_{P_6 G} \\
 \textstyle{1 \over 2}D_{SG} \\
 \end{array}} \right\} \le \frac{2}{5}D_{P_5 H},\\
\intertext{and}
\label{eq13}
& \left\{ {\begin{array}{l}
 \textstyle{1 \over 8}D_{P_6 P_1 } \\
 \textstyle{2 \over {13}}D_{P_5 P_1 } \\
 \end{array}} \right\} \le \left\{ {\begin{array}{l}
 \textstyle{1 \over 6}D_{P_6 P_2 } \\
 \textstyle{2 \over 9}D_{P_5 P_2 } \\
 \end{array}} \right\} \le \frac{2}{7}D_{P_5 P_3 } \le \frac{4}{9}D_{P_6 N_2
} \le D_{P_6 S} \le D_{AG}.
\end{align}

\noindent
where, for example, $D_{P_6 N_2 } = P_6 - N_2 $, $D_{AG} = A - G$, $D_{P_5 A} = P_5 - A$, etc.

\bigskip
\textbf{Notation:} \textit{Throughout the paper, the notation }$A \le \left\{ {\begin{array}{l}
 B \\
 C \\
 \end{array}} \right\}$\textit{ is understood as }$A \le B$\textit{ and }$A \le C$\textit{, but there is no relation between }$B$\textit{ and }$C.$

\bigskip
The aim of this paper is to improve the inequalities (\ref{eq8}) based on the
results appearing in the inequalities (\ref{eq11})-(\ref{eq13}).

\section{Refinement Inequalities}

The results appearing in the inequalities (\ref{eq11})-(\ref{eq13}) lead us to the
following two groups of individual inequalities:

\bigskip
\textbf{Group 1:}
\begin{multicols}{3}
\begin{align}
& 1. \hspace{15pt} P_2 \le \frac{P_6 + 3P_1 }{4}.\notag\\
& 2. \hspace{15pt} \frac{2S + H}{3} \le A.\notag\\
& 3. \hspace{15pt} \frac{P_6 + 14N_1 }{15} \le N_3.\notag\\
& 4. \hspace{15pt} \frac{P_6 + 2P_4 }{3} \le N_3.\notag\\
& 5. \hspace{15pt} S \le \frac{5P_6 + 2P_4 }{7}.\notag\\
& 6. \hspace{15pt} \frac{P_6 + 3G}{4} \le N_2.\notag\\
& 7. \hspace{15pt} \frac{P_6 + 5G}{6} \le N_1.\notag
\end{align}

\begin{align}
& 8. \hspace{15pt} G \le \frac{P_6 + 6P_4 }{7}.\notag\\
& 9. \hspace{15pt} \frac{2N_2 + G}{3} \le N_1.\notag\\
& 10. \hspace{15pt} N_2 \le \frac{3A + G}{4}.\notag\\
& 11. \hspace{15pt} N_1 \le \frac{P_5 + 2G}{3}.\notag\\
& 12. \hspace{15pt} N_3 \le \frac{P_5 + 5N_1 }{6}.\notag\\
& 13. \hspace{15pt} N_2 \le \frac{P_5 + 9N_3 }{10}.\notag\\
& 14. \hspace{15pt} A \le \frac{P_5 + 2N_2 }{3}. \notag
\end{align}
\begin{align}
& 15. \hspace{15pt} \frac{S + 4N_2 }{5} \le A.\notag\\
& 16. \hspace{15pt} \frac{S + 3N_3 }{4} \le A.\notag\\
& 17. \hspace{15pt} \frac{S + 5N_1 }{6} \le N_2.\notag\\
& 18. \hspace{15pt} \frac{S + 8N_1 }{9} \le N_3.\notag\\
& 19. \hspace{15pt} \frac{S + 3G}{4} \le N_1. \notag\\
& 20. \hspace{15pt} P_2 \le \frac{9P_1 + 4P_5 }{13}.\notag\\
& 21. \hspace{15pt} P_3 \le \frac{2P_5 + 7P_2 }{9}.\notag\\
& 22. \hspace{15pt} S \le \frac{5P_6 + 4N_2 }{9}.\notag
\end{align}
\end{multicols}

\begin{multicols}{2}
\textbf{Group 2:}
\begin{align}
& 1. \hspace{15pt} 6A + P_6 \le 6S + P_2.\notag\\
& 2. \hspace{15pt} 9A + 8N_2 \le 9H + 8P_6.\notag\\
& 3. \hspace{15pt} 7A + 6N_3 \le 7H + 6P_6.\notag\\
& 4. \hspace{15pt} 5A + 4P_4 \le 5H + 4S.\notag\\
& 6. \hspace{15pt} S + N_1 \le P_4 + P_6.\notag\\
& 6. \hspace{15pt} 6H + 5P_6 \le 6P_5 + 5G.\notag\\
& 7. \hspace{15pt} 2P_4 + P_6 \le 2A + G.\notag\\
& 8. \hspace{15pt} 10N_1 + P_5 \le 10N_2 + H.\notag
\end{align}

\begin{align}
& 9. \hspace{15pt} A + 6N_1 \le P_4 + 6N_2.\notag\\
& 10. \hspace{15pt} G + 6A \le P_5 + 6N_2\notag\\
& 11 \hspace{15pt} G + 2S \le P_6 + 2N_1.\notag\\
& 12. \hspace{15pt} 4H + 5S \le 4P_5 + 5G.\notag\\
& 13. \hspace{15pt} 16P_2 + 9P_6 \le 16P_5 + 9P_1.\notag\\
& 14. \hspace{15pt} 13P_2 + 12P_5 \le 13P_6 + 12P_1.\notag\\
& 15. \hspace{15pt} 12P_3 + 7P_6 \le 12P_5 + 7P_2.\notag\\
& 16. \hspace{15pt} 14N_2 + 9P_5 \le 14P_6 + 9P_3.\notag\\
& 17. \hspace{15pt} P_6 + G \le A + S.\notag
\end{align}
\end{multicols}

Based on the inequalities appearing in Group 1, we have the theorem giving
the refinement of the inequalities appearing in (\ref{eq8}).

\begin{theorem} The following inequalities hold:
\begin{align}
& G \le \textstyle{{P_6 + 6P_4 } \over 7} \le \left\{ {\begin{array}{l}
 \textstyle{{P_6 + 5G} \over 6} \\
 \textstyle{{S + 3G} \over 4} \\
 \end{array}} \right\} \le \textstyle{{2N_2 + G} \over 3} \le N_1 \le
\left\{ {\begin{array}{l}
 \textstyle{{10N_2 + H - P_5 } \over {10}} \\
 \textstyle{{P_4 + 6N_2 - A} \over 6} \\
 \end{array}} \right\} \le\notag\\
& \hspace{10pt} \le \left\{ {\begin{array}{l}
 \left\{ {\begin{array}{l}
 \textstyle{{P_5 + 2G} \over 3} \\
 P_4 + P_6 - S \\
 \end{array}} \right\} \le \textstyle{{P_6 + 2P_4 } \over 3} \\
 P_4 + P_6 - S\left\{ {\begin{array}{l}
 \textstyle{{P_6 + 14N_1 } \over {15}} \\
 \textstyle{{S + 8N_1 } \over 9} \\
 \end{array}} \right. \\
 \end{array}} \right\} \le N_3 \le\notag\\
& \hspace{20pt} \le N_2 \le \left\{ {\begin{array}{l}
 \textstyle{{3A + G} \over 4} \le \textstyle{{P_5 + 9N_3 } \over {10}} \le
\left\{ {\begin{array}{l}
 \textstyle{{8P_6 + 9H - 9A} \over 8} \\
 \textstyle{{S + 4N_2 } \over 5} \\
 \textstyle{{S + 3N_3 } \over 4} \\
 \textstyle{{2S + H} \over 3} \\
 \end{array}} \right\} \le A \\
 \textstyle{{14P_6 + 9P_3 - 9P_5 } \over {14}} \\
 \end{array}} \right\} \le \notag\\
& \hspace{30pt} \le \left\{ {\begin{array}{l}
 \textstyle{{5H + 4S - 4P_4 } \over 5} \\
 \textstyle{{P_5 + 6N_2 - G} \over 6} \le \textstyle{{P_5 + 2N_2 } \over 3}
\le \textstyle{{7H + 6P_6 - 6N_3 } \over 7} \\
 \end{array}} \right\} \le \notag\\
 & \hspace{45pt} \le \left\{ {\begin{array}{l}
 \textstyle{{4H + 5S - 5G} \over 4} \le P_5 \\
 S \le \textstyle{{P_6 + 2N_1 - G} \over 2} \le \left\{ {\begin{array}{l}
 \textstyle{{2P_4 + 5P_6 } \over 7} \\
 \textstyle{{4N_2 + 5P_6 } \over 9} \\
 \end{array}} \right. \\
 \textstyle{{12P_5 + 13P_2 - 12P_1 } \over {13}} \\
 \end{array}} \right\} \le P_6 \le \notag\\
 \label{eq14}
& \hspace{60pt} \le \left\{ {\begin{array}{l}
 A + S - G \\
 \textstyle{{6P_5 + 5G - 6H} \over 5} \\
 2A + G - 2P_4 \\
 \end{array}} \right\} \le \left\{ {\begin{array}{l}
 \textstyle{{12P_5 + 7P_2 - 12P_3 } \over 7} \\
 6S + P_2 - 6A \\
 \end{array}} \right\} \le \textstyle{{16P_5 + 9P_1 - 16P_2 } \over 9}
 \end{align}
\noindent and
\begin{equation}
\label{eq15}
P_2 \le \left\{ {\begin{array}{l}
 \textstyle{{P_6 + 3P_1 } \over 4} \le N_1 \le \textstyle{{P_6 + 3G} \over
4} \le \textstyle{{S + 5N_1 } \over 6} \le N_2 \\
 \textstyle{{4P_5 + 9P_1 } \over {13}} \le N_3 \le \textstyle{{P_5 + 5N_1 }
\over 6} \le A \le \textstyle{{3I + 2P_4 } \over 2} \le \left\{
{\begin{array}{l}
 P_5 \le \textstyle{{T + 2A} \over 2} \le \textstyle{{3J + 16G} \over {16}}
\\
 S \\
 \end{array}} \right. \\
 P_3 \le \textstyle{{2P_5 + 7P_2 } \over 9} \le N_1 \\
 \end{array}} \right..
\end{equation}
\end{theorem}

\begin{proof} Some of the results appearing in (\ref{eq14}) and (\ref{eq15}) are either
due to Groups 1 and 2 or are obvious. Here we shall prove only those are not
obvious.

\begin{enumerate}
\item \textbf{For }$\bf{\frac{P_6 + 6P_4 }{7} \le \frac{S + 3G}{4}}$\textbf{: }We have
to show that
\[
\textstyle{1 \over {28}}\left( {7S + 21G - 4P_6 - 24P_4 } \right) \ge 0.
\]

We can write $7S + 21G - 4P_6 - 24P_4 = \textstyle{1 \over 2}b\,g_1 \left(
{a \mathord{\left/ {\vphantom {a b}} \right. \kern-\nulldelimiterspace} b}
\right)$, where
\[
g_1 (x) = \frac{u_1 (x)}{\left( {x + 1} \right)\left( {\sqrt x + 1}
\right)^2},
\]

\noindent with
\begin{align}
u_1 (x) & = 7\sqrt {2x^2 + 2} \left( {\sqrt x + 1} \right)^2\left( {x + 1}
\right)\notag\\
& \hspace{20pt} - \left( {8x^3 - 26x^{5 / 2} + 116x^2 - 84x^{3 / 2} + 116x - 26\sqrt x + 8}
\right).\notag
\end{align}

Now we shall show that $u_1 (x) \ge 0$, $\forall x > 0$. Let us consider
\begin{align}
v_1 (x) & = \left[ {7\sqrt {2x^2 + 2} \left( {\sqrt x + 1} \right)^2\left( {x
+ 1} \right)} \right]^2\notag\\
& \hspace{20pt} - \left( {8x^3 - 26x^{5 / 2} + 116x^2 - 84x^{3 / 2} + 116x - 26\sqrt x + 8}
\right)^2\notag\\
& = 2\left( {\sqrt x - 1} \right)^2\left( {\begin{array}{l}
 8\sqrt x \left( {x^3 + 92x^2 + 92x + 1} \right)\left( {\sqrt x - 1}
\right)^2 + \\
 17x^5 + 430x^{9 / 2} + x^4 + 3064x^{7 / 2} + x + \\
 + 5522x^{5 / 2} + 3064x^{3 / 2} + 430\sqrt x + 17 \\
 \end{array}} \right).\notag
\end{align}

Since $v_1 (x) \ge 0$ giving $u_1 (x) \ge 0$, $\forall x > 0$, hence proving the required result.

\bigskip
\textbf{Argument:} \textit{Let }$a$\textit{ and }$b$\textit{ two positive numbers, i.e., }$a > 0$\textit{ and }$b > 0$\textit{. If }$a^2 - b^2 \ge 0$\textit{, then we can conclude that }$a \ge b$\textit{ because }$a - b = ({a^2 - b^2)} \mathord{\left/ {\vphantom {{a^2 - b^2)} {(a + b)}}} \right. \kern-\nulldelimiterspace} {(a + b)}$\textit{. If }$b < 0$\textit{, then obviously, }$a - ( - b) = a + b > 0$\textit{ always holds. In order to apply the above argument, it is sufficient that }$a >
0$\textit{. We have used this argument to prove }$u_1 (x) \ge 0, \forall x > 0$\textit{. We shall use frequently this argument to prove the other parts. }

\bigskip
\item \textbf{For }$\bf{\frac{S + 3G}{4} \le \frac{2N_2 + G}{3}}$\textbf{: }We have to show that
\[
\textstyle{1 \over {12}}\left( {8N_2 - 5G - 3S} \right) \ge 0.
\]

We can write $8N_2 - 5G - 3S = \textstyle{1 \over 2}b\,g_2 \left( {a \mathord{\left/ {\vphantom {a b}} \right. \kern-\nulldelimiterspace} b} \right)$, where
\[
g_2 (x) = 4\left( {\sqrt x + 1} \right)\sqrt {2x + 2} - \left( {10\sqrt x +
3\sqrt {2x^2 + 2} } \right).
\]

Now we shall show that $g_2 (x) \ge 0$, $\forall x > 0$. In order to prove it we shall apply twice the argument given in part 2 of section 3. Let us consider
\begin{align}
v_2 (x) & = \left[ {4\left( {\sqrt x + 1} \right)\sqrt {2x + 2} } \right]^2 -
\left( {10\sqrt x + 3\sqrt {2x^2 + 2} } \right)^2\notag\\
& = 14x^2 + 46x^{3 / 2} + 18\sqrt x \left( {\sqrt x - 1} \right)^2 + 46\sqrt
x + 14\notag\\
& \hspace{20pt} - 60\sqrt x \sqrt {2x^2 + 2}.\notag
\end{align}

Let us consider again
\begin{align}
v_{2a} (x) & = \left[ {14x^2 + 46x^{3 / 2} + 18\sqrt x \left( {\sqrt x - 1}
\right)^2 + 46\sqrt x + 14} \right]^2\notag\\
& \hspace{20pt} - \left[ {60\sqrt x \sqrt {2x^2 + 2} } \right]^2.\notag\\
& = 4\left( {\sqrt x - 1} \right)^4\left( {49x^2 + 644x^{3 / 2} + 1254x +
644\sqrt x + 49} \right).\notag
\end{align}

Since $v_{2a} (x) \ge 0$, giving $v_2 (x) \ge 0$, $\forall x > 0$. This implies that $g_2 (x) \ge 0$, $\forall x > 0$, hence proving the required result.

\bigskip
\item \textbf{For }$\bf{\frac{P_6 + 5G}{6} \le \frac{2N_2 + G}{3}}$\textbf{: }We have
to show that
\[
\textstyle{1 \over 6}\left( {4N_2 - 3G - P_6 } \right) \ge 0.
\]

We can write $4N_2 - 3G - P_6 = b\,g_3 \left( {a \mathord{\left/ {\vphantom {a b}} \right. \kern-\nulldelimiterspace} b} \right)$, where
\[
g_3 (x) = \frac{u_3 (x)}{x + 1},
\]

\noindent with
\[
u_3 (x) = \sqrt {2x + 2} \left( {\sqrt x + 1} \right)\left( {x + 1} \right)
- \left( {3x^{3 / 2} + 3\sqrt x + x^2 + 1} \right).
\]

Now we shall show that $u_3 (x) \ge 0$, $\forall x > 0$. Let us consider
\begin{align}
v_3 (x) & = \left[ {\sqrt {2x + 2} \left( {\sqrt x + 1} \right)\left( {x + 1}
\right)} \right]^2 - \left( {3x^{3 / 2} + 3\sqrt x + x^2 + 1} \right)^2\notag\\
& = \left( {\sqrt x - 1} \right)^2\left( {x^2 + 2x^{3 / 2} + 2\sqrt x + x +
1} \right).\notag
\end{align}

Since $v_3 (x) \ge 0$, giving $u_3 (x) \ge 0$, $\forall x > 0$, hence proving the required result.

\bigskip
\item \textbf{For }$\bf{\frac{10N_2 + H - P_5 }{10} \le P_4 + P_6 - S}$\textbf{: }We
have to show that
\[
\textstyle{1 \over {10}}\left( {10P_4 + 10P_6 \left. { + P_5 - 10N_2 - 10S -
H} \right) \ge 0} \right..
\]

We can write $10P_4 + 10P_6 + P_5 - 10N_2  \,  - 10S - H = \textstyle{1 \over
2}b\,g_4 \left( {a \mathord{\left/ {\vphantom {a b}} \right.
\kern-\nulldelimiterspace} b} \right)$, where
\[
g_4 (x) = \frac{u_4 (x)}{\left( {\sqrt x + 1} \right)^2\left( {x + 1}
\right)},
\]

\noindent with
\begin{align}
u_4 (x) & = 22x^3 + 40x^{5 / 2} + 98x^2 + 4x\left( {\sqrt x - 1} \right)^2 +
98x + 40\sqrt x + 22\notag\\
& \hspace{20pt} - 5\left( {\sqrt x + 1} \right)^2\left( {x + 1} \right)\left[ {2\sqrt {2x^2
+ 2} + 5\left( {\sqrt x + 1} \right)\sqrt {2x + 2} } \right].\notag
\end{align}

Now we shall show that $u_4 (x) \ge 0$, $\forall x > 0$. In order to prove
it we shall apply twice the argument given in Part 1. Let us consider
\begin{align}
v_4 (x) & = \left[ {22x^3 + 40x^{5 / 2} + 98x^2 + 4x\left( {\sqrt x - 1}
\right)^2 + 98x + 40\sqrt x + 22} \right]^2\notag\\
& \hspace{20pt} - \left\{ {5\left( {\sqrt x + 1} \right)^2\left( {x + 1} \right)\left[
{2\sqrt {2x^2 + 2} + 5\left( {\sqrt x + 1} \right)\sqrt {2x + 2} } \right]}
\right\}^2\notag\\
& = \left( {\begin{array}{l}
 234 + 17240x^3 + 888x^{5 / 2} + 8102x^2 + 3508x^{3 / 2} + \\
 + 3588x + 660\sqrt x + 234x^6 + 660x^{11 / 2} + \\
 + 3588x^5 + 3508x^{9 / 2} + 8102x^4 + 888x^{7 / 2} \\
 \end{array}} \right)\notag\\
& \hspace{20pt} - 100\sqrt {2x^2 + 2} \sqrt {2x + 2} \left( {\begin{array}{l}
 26x^2 + 12x^{7 / 2} + 5x^4 + \\
 + 26x^{5 / 2} + + 12x + 20x^3 + \\
 + 20x^{3 / 2} + x^{9 / 2} + 5\sqrt x + 1 \\
 \end{array}} \right).\notag
\end{align}

Let us consider again
\begin{align}
v_{4a} (x) & = \left( {\begin{array}{l}
 234 + 17240x^3 + 888x^{5 / 2} + 8102x^2 + 3508x^{3 / 2} + \\
 + 3588x + 660\sqrt x + 234x^6 + 660x^{11 / 2} + \\
 + 3588x^5 + 3508x^{9 / 2} + 8102x^4 + 888x^{7 / 2} \\
 \end{array}} \right)^2\notag\\
& \hspace{20pt} - \left[ {100\sqrt {2x^2 + 2} \sqrt {2x + 2} \left( {\begin{array}{l}
 26x^2 + 12x^{7 / 2} + 5x^4 + \\
 + 26x^{5 / 2} + + 12x + 20x^3 + \\
 + 20x^{3 / 2} + x^{9 / 2} + 5\sqrt x + 1 \\
 \end{array}} \right)} \right]^2\notag\\
& = 4\left( {\sqrt x - 1} \right)^4w_4 (x),\notag
\end{align}

\noindent where
\[
w_4 (x) = \left( {\begin{array}{l}
 3689 - 144760x^{3 / 2} - 8024x^{19 / 2} - 8024\sqrt x + \\
 + 402389x^2 + 6593432x^3 + 3689x^{10} - 25534x^9 + \\
 + 402389x^8 + 6593432x^7 + 17426946x^6 + \\
 + 26278348x^5 + 17426946x^4 + 1834912x^{15 / 2} + \\
 + 10215648x^{13 / 2} + 18146128x^{11 / 2} + 18146128x^{9 / 2} + \\
 + 10215648x^{7 / 2} - 25534x + 1834912x^{5 / 2} - 144760x^{17 / 2} \\
 \end{array}} \right).
\]

Now we shall show that $w_4 (x) > 0$, $\forall x > 0$. Let us consider
\[
h_4 (t) = w_4 (t^2) = \left( {\begin{array}{l}
 3689t^{20} - 8024t^{19} - 25534t^{18} - 144760t^{17} + \\
 + 402389t^{16} + 1834912t^{15} + 6593432t^{14} + \\
 + 10215648t^{13} + 17426946t^{12} + 18146128t^{11} + \\
 + 26278348t^{10} + 18146128t^9 + 17426946t^8 + \\
 + 10215648t^7 + 6593432t^6 + 1834912t^5 + \\
 + 402389t^4 - 144760t^3 - 25534t^2 - 8024t + 3689 \\
 \end{array}} \right).
\]

The polynomial equation $h_4 (t) = 0$ of 20$^{th}$ degree admits 20 solutions. All of them are convex (not written here). This means that there are no real positive solutions of the equation $h_4 (t) = 0$. Thus we conclude that either $h_4 (t) > 0$ or $h_4 (t) < 0$, for all $t > 0$. In order to check it is sufficient to see for any particular value of $h_4 (t)$, for example when $t = 1$. This gives $h_4 (1) = \mbox{135168000}$, hereby proving that $h_4 (t) > 0$ for all $t > 0$, consequently, $v_{4a} (x)
\ge 0$, for all $x > 0$ giving $v_4 (x) \ge 0$, for all $x > 0$. Finally, proving the required result.

\bigskip
\item \textbf{For }$\bf{\frac{P_4 + 6N_2 - A}{6} \le P_4 + P_6 - S}$\textbf{: }We have
to show that
\[
\textstyle{1 \over 6}\left( {5P_4 + 6P_6 + A - 6S - 6N_2 } \right) \ge 0.
\]

We can write $5P_4 + 6P_6 + A - 6S - 6N_2 = \textstyle{1 \over 2}b\,g_5 \left( {a \mathord{\left/ {\vphantom {a b}} \right. \kern-\nulldelimiterspace} b} \right)$, where
\[
g_5 (x) = \frac{u_5 (x)}{\left( {\sqrt x + 1} \right)^2\left( {x + 1}
\right)},
\]

\noindent with
\begin{align}
u_5 (x) & = 13x^3 + 26x^{5 / 2} + 55x^2 + 4x^{3 / 2} + 55x + 26\sqrt x + 13\notag\\
& \hspace{20pt} - 3\left( {x + 1} \right)\left( {\sqrt x + 1} \right)^2\left[ {2\sqrt {2x^2
+ 2} + \left( {\sqrt x + 1} \right)\sqrt {2x + 2} } \right].\notag
\end{align}

Now we shall show that $u_5 (x) \ge 0$, $\forall x > 0$. Again in this case also we shall apply twice the argument given in Part 1. Let us consider
\begin{align}
v_5 (x) & = \left( {13x^3 + 26x^{5 / 2} + 55x^2 + 4x^{3 / 2} + 55x + 26\sqrt x
+ 13} \right)^2\notag\\
& \hspace{20pt} - \left\{ {3\left( {x + 1} \right)\left( {\sqrt x + 1} \right)^2\left[
{2\sqrt {2x^2 + 2} + \left( {\sqrt x + 1} \right)\sqrt {2x + 2} } \right]}
\right\}^2\notag\\
& = \left( {\begin{array}{l}
 79 + 4948x^3 + 1312x^{5 / 2} + 2449x^2 + 1416x^{3 / 2} + \\
 + 1206x + 280\sqrt x + 2449x^4 + 1312x^{7 / 2} + \\
 + 79x^6 + 280x^{11 / 2} + 1206x^5 + 1416x^{9 / 2} \\
 \end{array}} \right)\notag\\
& \hspace{20pt} - 36\sqrt {2x^2 + 2} \sqrt {2x + 2} \left( {\begin{array}{l}
 x^{9 / 2} + 5x^4 + 12x^{7 / 2} + \\
 + 20x^3 + 26x^{5 / 2} + 26x^2 + \\
 + 20x^{3 / 2} + 12x + 5\sqrt x + 1 \\
 \end{array}} \right).\notag
\end{align}

Let us consider again
\begin{align}
v_{5a} (x) & = \left( {\begin{array}{l}
 79 + 4948x^3 + 1312x^{5 / 2} + 2449x^2 + 1416x^{3 / 2} + \\
 + 1206x + 280\sqrt x + 2449x^4 + 1312x^{7 / 2} + \\
 + 79x^6 + 280x^{11 / 2} + 1206x^5 + 1416x^{9 / 2} \\
 \end{array}} \right)^2\notag\\
& \hspace{20pt} - \left[ {36\sqrt {2x^2 + 2} \sqrt {2x + 2} \left( {\begin{array}{l}
 x^{9 / 2} + 5x^4 + 12x^{7 / 2} + \\
 + 20x^3 + 26x^{5 / 2} + 26x^2 + \\
 + 20x^{3 / 2} + 12x + 5\sqrt x + 1 \\
 \end{array}} \right)} \right]^2\notag\\
& = \left( {\sqrt x - 1} \right)^4w_5 (x),\notag
\end{align}

\noindent where
\[
w_5 (x) = \left( {\begin{array}{l}
 1057 - 3372\sqrt x - 10082x + 375981x^2 + \\
 + 3864616x^3 + 9785650x^4 + 1057x^{10} - 10082x^9 + \\
 + 375981x^8 + 3864616x^7 + 9785650x^6 + \\
 + 13943412x^5 + 1940x^{3 / 2} + 1462448x^{5 / 2} - 3372x^{19 / 2} + \\
 + 1940x^{17 / 2} + 1462448x^{15 / 2} + 6489648x^{13 / 2} + \\
 + 11758264x^{11 / 2} + 11758264x^{9 / 2} + 6489648x^{7 / 2} \\
 \end{array}} \right).
\]

Now we shall show that $w_5 (x) > 0$, $\forall x > 0$. Let us consider
\[
h_5 (t) = w_5 (t^2) = \left( {\begin{array}{l}
 1057t^{20} - 3372t^{19} - 10082t^{18} + 1940t^{17} + \\
 + 375981t^{16} + 1462448t^{15} + 3864616t^{14} + \\
 + 6489648t^{13} + 9785650t^{12} + 11758264t^{11} + \\
 + 13943412t^{10} + 11758264t^9 + 9785650t^8 + \\
 + 6489648t^7 + 3864616t^6 + 1462448t^5 + \\
 + 375981t^4 + 1940t^3 - 10082t^2 - 3372t + 1057 \\
 \end{array}} \right).
\]

The polynomial equation $h_5 (t) = 0$ of 20$^{th}$ degree admits 20 solutions. All of them are convex (not written here). This means that there are no real positive solutions of the equation $h_5 (t) = 0$. Thus we conclude that either $h_5 (t) > 0$ or $h_5 (t) < 0$, for all $t > 0$. Since $h_5 (1) = \mbox{81395712}$, this gives that $h_{17} (t) > 0$ for all $t > 0$, consequently, $v_{5a} (x) \ge 0$, for all $x > 0$ giving $v_5 (x) \ge 0$, for all $x > 0$. Finally, we have the required result.

\bigskip
\item \textbf{For }$\bf{\frac{10N_2 + H - P_5 }{10} \le \frac{P_5 + 2G}{3}}$\textbf{:
}We have to show that
\[
\textstyle{1 \over {30}}\left( {13P_5 + 20G - 30N_2 - 3H} \right) \ge 0.
\]

We can write $13P_5 + 20G - 30N_2 - 3H = \textstyle{1 \over 2}b\,g_6 \left( {a \mathord{\left/ {\vphantom {a b}} \right. \kern-\nulldelimiterspace} b} \right)$, where
\[
g_6 (x) = \frac{u_6 (x)}{\left( {\sqrt x + 1} \right)^2\left( {x + 1}
\right)},
\]

\noindent with
\begin{align}
u_6 (x) & = 26x^3 + 40x^{5 / 2} + 146x^2 + 56x^{3 / 2} + 146x + 40\sqrt x + 26\notag\\
& \hspace{20pt} - 15\left( {x + 1} \right)\left( {\sqrt x + 1} \right)^3\sqrt {2x + 2}.\notag
\end{align}

Now we shall show that $u_6 (x) \ge 0$, $\forall x > 0$. Let us consider
\begin{align}
v_6 (x) & = \left( {26x^3 + 40x^{5 / 2} + 146x^2 + 56x^{3 / 2} + 146x +
40\sqrt x + 26} \right)^2\notag\\
& \hspace{20pt} - \left[ {15\left( {x + 1} \right)\left( {\sqrt x + 1} \right)^3\sqrt {2x +
2} } \right]^2.\notag\\
& = 2\left( {\sqrt x - 1} \right)^4\left( {\begin{array}{l}
 113x^4 + 142x^{7 / 2} + 436x^3 + 90x^{5 / 2} + \\
 + 718x^2 + 90x^{3 / 2} + 436x + 142\sqrt x + 113 \\
 \end{array}} \right).\notag
\end{align}

Since $v_6 (x) \ge 0$, giving $u_6 (x) \ge 0$, $\forall x > 0$, hence proving the required result.

\bigskip
\item \textbf{For }$\bf{\frac{P_6 + 6N_2 - A}{6} \le \frac{P_5 + 2G}{3}}$\textbf{: }We
have to show that
\[
\textstyle{1 \over 6}\left( {2P_5 + 4G + A - P_4 - 6N_2 } \right) \ge 0.
\]

We can write $2P_5 + 4G + A - P_4 - 6N_2 = \textstyle{1 \over 2}b\,g_7 \left( {a \mathord{\left/ {\vphantom {a b}} \right. \kern-\nulldelimiterspace} b} \right)$, where
\[
g_6 (x) = \frac{u_7 (x)}{\left( {\sqrt x + 1} \right)^2},
\]

\noindent with
\[
u_7 (x) = 5x^2 + 10x^{3 / 2} + 18x + 10\sqrt x + 5
 - 3\left( {\sqrt x + 1} \right)^3\sqrt {2x + 2} .
\]

Now we shall show that $u_6 (x) \ge 0$, $\forall x > 0$. Let us consider
\begin{align}
v_7 (x) &= \left( {5x^2 + 10x^{3 / 2} + 18x + 10\sqrt x + 5} \right)^2
 - \left[ {3\left( {\sqrt x + 1} \right)^3\sqrt {2x + 2} } \right]^2\notag\\
& = \left( {\sqrt x - 1} \right)^4\left( {7x^2 + 20x^{3 / 2} + 30x + 20\sqrt
x + 7} \right).\notag
\end{align}

Since $v_7 (x) \ge 0$, giving $u_7 (x) \ge 0$, $\forall x > 0$, hence proving the required result.

\bigskip
\item \textbf{For }$\bf{P_4 + P_6 - S \le \frac{P_6 + 2P_4 }{3}}$\textbf{: }We have to
show that
\[
\textstyle{1 \over 3}\left( {3S - 2P_6 - P_4 } \right) \ge 0.
\]

We can write $3S - 2P_6 - P_4 = \textstyle{1 \over 2}b\,g_8 \left( {a \mathord{\left/ {\vphantom {a b}} \right. \kern-\nulldelimiterspace} b} \right)$, where
\[
g_8 (x) = \frac{u_8 (x)}{\left( {x + 1} \right)\left( {\sqrt x + 1}
\right)^2},
\]

\noindent with
\begin{align}
u_8 (x) & = 4x^3 + 8x^{5 / 2} + 12x^2 + 12x + 8\sqrt x + 4\notag\\
& \hspace{20pt} - 3\left( {x + 1} \right)\left( {\sqrt x + 1} \right)^2\sqrt {2x^2 + 2}.\notag
\end{align}

Now we shall show that $u_8 (x) \ge 0$, $\forall x > 0$. Let us consider
\begin{align}
v_8 (x) & = \left( {4x^3 + 8x^{5 / 2} + 12x^2 + 12x + 8\sqrt x + 4} \right)^2\notag\\
& \hspace{20pt} - \left[ {3\left( {x + 1} \right)\left( {\sqrt x + 1} \right)^2\sqrt {2x^2
+ 2} } \right]^2.\notag\\
& = 2\left( {\sqrt x - 1} \right)^2\left( {\begin{array}{l}
 x^5 + 6x^{9 / 2} + 3x^4 + 12x^{7 / 2} + 36x^3 + \\
 + 76x^{5 / 2} + 36x^2 + 12x^{3 / 2} + 3x + 6\sqrt x + 1 \\
 \end{array}} \right).\notag
\end{align}

Since $v_8 (x) \ge 0$, giving $u_8 (x) \ge 0$, $\forall x > 0$, hence proving the required result.

\bigskip
\item \textbf{For }$\bf{P_4 + P_6 - S \le \frac{P_6 + 14N_1 }{15}}$\textbf{: }We have
to show that
\[
\textstyle{1 \over {15}}\left( {15S + 14N_1 - 14P_6 - 15P_4 } \right) \ge 0.
\]

We can write $15S + 14N_1 - 14P_6 - 15P_4 = \textstyle{1 \over 2}b\,g_9 \left( {a \mathord{\left/ {\vphantom {a b}} \right. \kern-\nulldelimiterspace} b} \right)$, where
\[
g_9 (x) = \frac{u_9 (x)}{\left( {x + 1} \right)\left( {\sqrt x + 1}
\right)^2},
\]

\noindent with
\begin{align}
u_9 (x) & = 15\left( {x + 1} \right)\left( {\sqrt x + 1} \right)^2\sqrt {2x^2
+ 2}\notag\\
& \hspace{20pt} - \left( {21x^3 + 28\sqrt x \left( {x - 1} \right)^2 + 99x^2 + 99x + 21}
\right).\notag
\end{align}

Now we shall show that $u_9 (x) \ge 0$, $\forall x > 0$. Let us consider
\begin{align}
u_9 (x) & = \left[ {15\left( {x + 1} \right)\left( {\sqrt x + 1}
\right)^2\sqrt {2x^2 + 2} } \right]^2\notag\\
& \hspace{20pt} - \left( {21x^3 + 28\sqrt x \left( {x - 1} \right)^2 + 99x^2 + 99x + 21}
\right)^2\notag\\
& = \left( {\sqrt x - 1} \right)^2\left( {\begin{array}{l}
 \sqrt x \left( {34x^3 + 571x^2 + 571x + 34} \right)\left( {\sqrt x - 1}
\right)^2 + \\
 + 9x^5 + 608x^{9 / 2} + x^4 + 827x^{7 / 2} + 6710x^{5 / 2} + \\
 + 827x^{3 / 2} + x + 608\sqrt x + 9 \\
 \end{array}} \right). \notag
\end{align}

Since $v_9 (x) \ge 0$, giving $u_9 (x) \ge 0$, $\forall x > 0$, hence proving the required result.

\bigskip
\item \textbf{For }$\bf{P_4 + P_6 - S \le \frac{S + 8N_1 }{9}}$\textbf{: }We have to
show that
\[
\textstyle{1 \over 9}\left( {10S + 8N_1 - 9P_6 - 9P_4 } \right) \ge 0.
\]

We can write $10S + 8N_1 - 9P_6 - 9P_4 = b\,g_{10} \left( {a \mathord{\left/ {\vphantom {a b}} \right. \kern-\nulldelimiterspace} b} \right)$, where
\[
g_{10} (x) = \frac{u_{10} (x)}{\left( {x + 1} \right)\left( {\sqrt x + 1}
\right)^2},
\]

\noindent with
\begin{align}
u_{10} (x) & = 5\left( {x + 1} \right)\left( {\sqrt x + 1} \right)^2\sqrt
{2x^2 + 2}\notag\\
& \hspace{20pt} - \left[ {7x^3 + 10x^{5 / 2} + 23x^2 + 8x\left( {\sqrt x - 1} \right)^2 +
23x + 10\sqrt x + 7} \right].\notag
\end{align}

Now we shall show that $u_{10} (x) \ge 0$, $\forall x > 0$. Let us consider
\begin{align}
u_{10} (x) & = \left[ {5\left( {x + 1} \right)\left( {\sqrt x + 1}
\right)^2\sqrt {2x^2 + 2} } \right]^2\notag\\
& \hspace{20pt} - \left[ {7x^3 + 10x^{5 / 2} + 23x^2 + 8x\left( {\sqrt x - 1} \right)^2 +
23x + 10\sqrt x + 7} \right]^2\notag\\
&= \left( {\sqrt x - 1} \right)^2\left( {\begin{array}{l}
 \sqrt x \left( {6x^3 + 37x^2 + 37x + 6} \right)\left( {\sqrt x - 1}
\right)^2 + \\
 + x^5 + 56x^{9 / 2} + x^4 + 71x^{7 / 2} + 690x^{5 / 2} + \\
 + 71x^{3 / 2} + 56\sqrt x + x + 1 \\
 \end{array}} \right). \notag
\end{align}

Since $v_{10} (x) \ge 0$, giving $u_{10} (x) \ge 0$, $\forall x > 0$, hence proving the required result.

\bigskip
\item \textbf{For }$\bf{\frac{P_5 + 2G}{3} \le \frac{P_6 + 2P_4 }{3}}$\textbf{:} We
have to show that
\[
\textstyle{1 \over 3}\left( {P_6 + 2P_4 - P_5 - 2G} \right) \ge 0.
\]

We can write $P_6 + 2P_4 - P_5 - 2G = b\,g_{11} \left( {a \mathord{\left/ {\vphantom {a b}} \right. \kern-\nulldelimiterspace} b} \right)$, where
\[
g_{11} (x) = \frac{x\left( {\sqrt x - 1} \right)^2}{\left( {x + 1}
\right)\left( {x^{3 / 2} + 1} \right)\left( {\sqrt x + 1} \right)}.
\]

Since $g_{11} (x) \ge 0, \, \forall x > 0$, hence proving the required result.

\bigskip
\item \textbf{For }$\bf{\frac{3A + G}{4} \le \frac{P_5 + 9N_3 }{10}}$\textbf{:} We have
to show that
\[
\textstyle{1 \over {20}}\left( {2P_5 + 18N_3 - 15A - 5G} \right) \ge 0.
\]

We can write $2P_5 + 18N_3 - 15A - 5G = \textstyle{1 \over 2}b\,g_{12} \left( {a \mathord{\left/ {\vphantom {a b}} \right. \kern-\nulldelimiterspace} b} \right)$, where
\[
g_{12} (x) = \frac{\left( {\sqrt x - 1} \right)^4}{\left( {\sqrt x + 1}
\right)^2}.
\]

Since $g_{12} (x) \ge 0, \, \forall x > 0$, hence proving the required result.

\bigskip
\item \textbf{For }$\bf{\frac{P_5 + 9N_3 }{10} \le \frac{S + 4N_2 }{5}}$\textbf{: }We
have to show that
\[
\textstyle{1 \over {10}}\left( {2S + 8N_2 - P_5 - 9N_3 } \right) \ge 0.
\]

We can write $2S + 8N_2 - P_5 - 9N_3 = b\,g_{13} \left( {a \mathord{\left/ {\vphantom {a b}} \right. \kern-\nulldelimiterspace} b} \right)$, where
\[
g_{13} (x) = \frac{u_{13} (x)}{\left( {\sqrt x + 1} \right)^2},
\]

\noindent with
\begin{align}
u_{13} (x) & = \left( {\sqrt x + 1} \right)^2\left[ {\sqrt {2x^2 + 2} + 2\sqrt
{2x + 2} \left( {\sqrt x + 1} \right)} \right]\notag\\
& \hspace{20pt} - \left( {4x^2 + 9x^{3 / 2} + 14x + 9\sqrt x + 4} \right).\notag
\end{align}

Now we shall show that $u_{13} (x) \ge 0$, $\forall x > 0$. In order to prove it we shall apply twice the argument given in Part 1. Let us consider
\begin{align}
v_{13} (x) & = \left( {\sqrt x + 1} \right)\left[ {\sqrt {2x^2 + 2} + 2\sqrt
{2x + 2} \left( {\sqrt x + 1} \right)} \right]^2\notag\\
& \hspace{20pt} - \left( {4x^2 + 9x^{3 / 2} + 14x + + 9\sqrt x + 4} \right)^2\notag\\
& = 4\sqrt {2x^2 + 2} \sqrt {2x + 2} \left( {\sqrt x + 1} \right)^5\notag\\
& \hspace{20pt} - \left( {\begin{array}{l}
 6x^4 + 16x^{7 / 2} + 53x^3 + 108x^{5 / 2} + \\
 + 146x^2 + 108x^{3 / 2} + 53x + 16\sqrt x + 6 \\
 \end{array}} \right).\notag
\end{align}

Let us consider again
\begin{align}
v_{13a} (x) & = \left[ {4\sqrt {2x^2 + 2} \sqrt {2x + 2} \left( {\sqrt x + 1}
\right)^5} \right]^2\notag\\
& \hspace{20pt} - \left( {\begin{array}{l}
 6x^4 + 16x^{7 / 2} + 53x^3 + 108x^{5 / 2} + \\
 + 146x^2 + 108x^{3 / 2} + 53x + 16\sqrt x + 6 \\
 \end{array}} \right)^2\notag\\
& = \left( {\sqrt x - 1} \right)^2\left( {\begin{array}{l}
 28 + 504\sqrt x + 3032x + 27111x^2 + 72213x^3 + \\
 + 27111x^5 + 3032x^6 + 28x^7 + 72213x^4 + \\
 + 10888x^{3 / 2} + 50366x^{5 / 2} + 504x^{13 / 2} + \\
 + 10888x^{11 / 2} + 50366x^{9 / 2} + 81316x^{7 / 2} \\
 \end{array}} \right).\notag
\end{align}

Since $v_{13a} (x) \ge 0$, giving $v_{13} (x) \ge 0$, $\forall x > 0$. This implies that $u_{13} (x) \ge 0$, $\forall x > 0$, hence proving the required result.

\bigskip
\item \textbf{For }$\bf{\frac{P_5 + 9N_3 }{10} \le \frac{S + 3N_3 }{4}}$\textbf{: }We
have to show that
\[
\textstyle{1 \over {20}}\left( {5S - 3N_3 - 2P_5 } \right) \ge 0.
\]

We can write $5S - 3N_3 - 2P_5 = \textstyle{1 \over 2}b\,g_{14} \left( {a \mathord{\left/ {\vphantom {a b}} \right. \kern-\nulldelimiterspace} b} \right)$, where
\[
g_{14} (x) = \frac{u_{14} (x)}{\left( {\sqrt x + 1} \right)^2},
\]

\noindent with
\[
u_{14} (x) = 5\sqrt {2x^2 + 2} \left( {\sqrt x + 1} \right)^2 - 2\left(
{3x^2 + 3x^{3 / 2} + 8x + 3\sqrt x + 3} \right).
\]

Now we shall show that $u_{14} (x) \ge 0$, $\forall x > 0$. Let us consider
\begin{align}
v_{14} (x) & = \left[ {5\sqrt {2x^2 + 2} \left( {\sqrt x + 1} \right)^2}
\right]^2 - \left[ {2\left( {3x^2 + 3x^{3 / 2} + 8x + 3\sqrt x + 3} \right)}
\right]^2\notag\\
& = 2\left( {\sqrt x - 1} \right)^2\left( {\begin{array}{l}
 7x^3 + 78x^{5 / 2} + 185x^2 + \\
 + 260x^{3 / 2} + 185x + 78\sqrt x + 7 \\
 \end{array}} \right).\notag
\end{align}

Since $v_{14} (x) \ge 0$, giving $u_{14} (x) \ge 0$, $\forall x > 0$, hence proving the required result.

\bigskip
\item \textbf{For }$\bf{\frac{P_5 + 9N_3 }{10} \le \frac{2S + H}{3}}$\textbf{: }We have
to show that
\[
\textstyle{1 \over {30}}\left( {20S + 10H - 3P_5 - 27N_3 } \right) \ge 0.
\]

We can write $20S + 10H - 3P_5 - 27N = b\,g_{15} \left( {a \mathord{\left/ {\vphantom {a b}} \right. \kern-\nulldelimiterspace} b} \right)$, where
\[
g_{15} (x) = \frac{u_8 (x)}{\left( {x + 1} \right)\left( {\sqrt x + 1}
\right)^2},
\]

\noindent with
\begin{align}
u_{15} (x) & = 10\sqrt {2x^2 + 2} \left( {\sqrt x + 1} \right)^2\left( {x + 1}
\right)\notag\\
& \hspace{20pt}  - \left( {12x^3 + 27x^{5 / 2} + 14x^{3 / 2} + 34x^2 + 34x + 27\sqrt x + 12}
\right).\notag
\end{align}

Now we shall show that $u_{15} (x) \ge 0$, $\forall x > 0$. Let us consider
\begin{align}
v_{15} (x) & = \left[ {10\sqrt {2x^2 + 2} \left( {\sqrt x + 1} \right)^2\left(
{x + 1} \right)} \right]^2\notag\\
& \hspace{20pt}  - \left( {12x^3 + 27x^{5 / 2} + 14x^{3 / 2} + 34x^2 + 34x + 27\sqrt x + 12}
\right)^2\notag\\
& = \left( {\sqrt x - 1} \right)^2\left( {\begin{array}{l}
 56x^5 + 264x^{9 / 2} + 527x^4 + 1018x^{7 / 2} + \\
 + 1781x^3 + 2308x^{5 / 2} + 1781x^2 + \\
 + 1018x^{3 / 2} + 527x + 264\sqrt x + 56 \\
 \end{array}} \right).\notag
\end{align}

Since $v_{15} (x) \ge 0$, giving $u_{15} (x) \ge 0$, $\forall x > 0$, hence proving the required result.

\bigskip
\item \textbf{For }$\bf{\frac{P_5 + 9N_3 }{10} \le \frac{8P_6 + 9H - 9A}{8}}$\textbf{:
}We have to show that
\[
\textstyle{1 \over {40}}\left( {40P_6 + 45H - 45A - 4P_5 - 36N_3 } \right)
\ge 0.
\]

We can write $40P_6 + 45H - 45A - 4P_5 - 36N_3 = \textstyle{1 \over 2}b\,g_{16} \left( {a \mathord{\left/ {\vphantom {a b}} \right. \kern-\nulldelimiterspace} b} \right)$, where
\[
g_{16} (x) = \frac{\left( {\sqrt x + 3} \right)\left( {3\sqrt x + 1}
\right)\left( {\sqrt x - 1} \right)^4}{\left( {x + 1} \right)\left( {\sqrt x
+ 1} \right)^2}.
\]

Since $g_{16} (x) \ge 0$, $\forall x > 0$, hence proving the required result.

\bigskip
\item \textbf{For }$\bf{\frac{8P_6 + 9H - 9A}{8} \le A}$\textbf{: }We have to show that
\[
\textstyle{1 \over 8}\left( {17A - 8P_6 - 9H} \right) \ge 0.
\]

We can write $17A - 8P_6 - 9H = \textstyle{1 \over 2}b\,g_{17} \left( {a \mathord{\left/ {\vphantom {a b}} \right. \kern-\nulldelimiterspace} b} \right)$, where
\[
g_{17} (x) = \frac{\left( {x - 1} \right)^2}{\left( {x + 1} \right)}.
\]

Obviously, $g_{17} (x) > 0,\,\forall x > 0$. This proves the required result.

\bigskip
\item \textbf{For }$\bf{\frac{14P_6 + 9P_3 - 9P_5 }{14} \le \frac{5H + 4S - 4P_4
}{5}}$\textbf{: }We have to show that
\[
\textstyle{1 \over {70}}\left( {45P_5 + 70H + 56S - 56P_4 - 70P_6 - 45P_3 }
\right) \ge 0.
\]

We can write $45P_5 + 70H + 56S - 56P_4  \,  - 70P_6 - 45P_3 = b\,g_{18} \left( {a \mathord{\left/ {\vphantom {a b}} \right. \kern-\nulldelimiterspace} b} \right)$, where
\[
g_{18} (x) = \frac{u_{18} (x)}{\left( {\sqrt x + 1} \right)\left( {x^{3 / 2}
+ 1} \right)\left( {x + 1} \right)},
\]

\noindent with
\begin{align}
u_{18} (x) & = 28\sqrt {2x^2 + 2} \left( {x + 1} \right)\left( {x^{3 / 2} + 1}
\right)\left( {\sqrt x + 1} \right)\notag\\
& \hspace{20pt}  - \left( {\begin{array}{l}
 25x^4 + 115x^{7 / 2} - 51x^3 - 69x^{5 / 2} + \\
 + 408x^2 - 69x^{3 / 2} - 51x + 115\sqrt x + 25 \\
 \end{array}} \right).\notag
\end{align}

Now we shall show that $u_{18} (x) \ge 0$, $\forall x > 0$. Let us consider
\begin{align}
v_{18} (x) & = 28\sqrt {2x^2 + 2} \left( {x + 1} \right)\left( {x^{3 / 2} + 1}
\right)\left( {\sqrt x + 1} \right)\notag\\
& \hspace{20pt}  - \left( {\begin{array}{l}
 25x^4 + 115x^{7 / 2} - 51x^3 - 69x^{5 / 2} + \\
 + 408x^2 - 69x^{3 / 2} - 51x + 115\sqrt x + 25 \\
 \end{array}} \right)\notag\\
& = \left( {\sqrt x - 1} \right)^2w_{18} (x),\notag
\end{align}

\noindent where
\[
w_{18} (x) = \left( {\begin{array}{l}
 943x^7 - 728x^{13 / 2} - 8370x^6 + 8576x^{11 / 2} + \\
 + 30935x^5 - 28454x^{9 / 2} - 12184x^4 + \\
 + 81284x^{7 / 2} - 12184x^3 - 28454x^{5 / 2} \\
 + 30935x^2 + 8576x^{3 / 2} - 8370x - 728\sqrt x + 943 \\
 \end{array}} \right).
\]

Now we shall show that $w_{18} (x) > 0$, $\forall x > 0$. Let us consider
\[
h_{18} (t) = w_{18} (t^2) = \left( {\begin{array}{l}
 943t^{14} - 728t^{13} - 8370t^{12} + 8576t^{11} + \\
 + 30935t^{10} - 28454t^9 - 12184t^8 + \\
 + 81284t^7 - 12184t^6 - 28454t^5 + \\
 + 30935t^4 + 8576t^3 - 8370t^2 - 728t + 943 \\
 \end{array}} \right).
\]

The polynomial equation $h_{18} (t) = 0$ of 14$^{th}$ degree admits 14 solutions. Out of them 12 are convex (not written here) and two of them are real given by $ - 1.566438336$ and $ - 0.6383909134$. These two real solutions are negative. Since we are working with $t > 0$, this means that there are no real positive solutions of the equation $h_{18} (t) = 0$. Since, $h_{18} (1) = \mbox{62720}$, this gives that $h_{18} (t) > 0$, for all $t > 0$. Thus we have $w_{18} (x) > 0$, for all $x > 0$ giving $v_{18} (x) \ge 0$, for all $x > 0$. Finally, proving the required result.

\bigskip
\item \textbf{For }$\bf{\frac{14P_6 + 9P_3 - 9P_5 }{14} \le \frac{P_5 + 6N_2 -
G}{6}}$\textbf{: }Equivalently, we have to show that
\[
\textstyle{1 \over {42}}\left( {34P_5 + 42N_2 - 42P_6 - 27P_3 - 7G} \right)
\ge 0.
\]

We can write $34P_5 + 42N_2 - 42P_6 - 27P_3  \,  - 7G = \textstyle{1 \over 2}b\,g_{19} \left( {a \mathord{\left/ {\vphantom {a b}} \right. \kern-\nulldelimiterspace} b} \right)$, where
\[
g_{19} (x) = \frac{u_{19} (x)}{\left( {\sqrt x + 1} \right)\left( {x^{3 / 2}
+ 1} \right)\left( {x + 1} \right)},
\]

\noindent with
\begin{align}
u_{19} (x) & = 21\sqrt {2x + 2} \left( {x + 1} \right)\left( {x^{3 / 2} + 1}
\right)\left( {\sqrt x + 1} \right)^2\notag\\
& \hspace{20pt}  - \left( {\begin{array}{l}
 16x^4 + 64x^{7 / 2} + 256x^{5 / 2} + 256x^{3 / 2} + 64\sqrt x + \\
 + 2\sqrt x \left( {51x^2 + 26x + 51} \right)\left( {\sqrt x - 1} \right)^2
+ 16 \\
 \end{array}} \right)\notag
\end{align}

Now we shall show that $u_{19} (x) \ge 0$, $\forall x > 0$. Let us consider
\begin{align}
v_{19} (x) & = \left[ {21\sqrt {2x + 2} \left( {x + 1} \right)\left( {x^{3 /
2} + 1} \right)\left( {\sqrt x + 1} \right)^2} \right]^2 \notag\\
& \hspace{20pt}  - \left( {\begin{array}{l}
 16x^4 + 64x^{7 / 2} + 256x^{5 / 2} + 256x^{3 / 2} + 64\sqrt x + \\
 + 2\sqrt x \left( {51x^2 + 26x + 51} \right)\left( {\sqrt x - 1} \right)^2
+ 16 \\
 \end{array}} \right)^2\notag\\
& = 2\left( {\sqrt x - 1} \right)^2w_{19} (x),\notag
\end{align}

\noindent where
\[
w_{19} (x) = \left( {\begin{array}{l}
 313x^7 - 266x^{13 / 2} - 7390x^6 + 20728x^{11 / 2} - \\
 - 25128x^5 + 41882x^{9 / 2} - 36915x^4 + \\
 + 70000x^{7 / 2} - 36915x^3 + 41882x^{5 / 2} - \\
 - 25128x^2 + 20728x^{3 / 2} - 7390x - 266\sqrt x + 313 \\
 \end{array}} \right).
\]

Now we shall show that $w_{19} (x) > 0$, $\forall x > 0$. Let us consider
\[
h_{19} (t) = w_{19} (t^2) = \left( {\begin{array}{l}
 313t^{14} - 266t^{13} - 7390t^{12} + 20728t^{11} - 25128t^{10} + \\
 + 41882t^9 - 36915t^8 + 70000t^7 - 36915t^6 + \\
 + 41882t^5 - 25128t^4 + 20728t^3 - 7390t^2 - 266t + 313 \\
 \end{array}} \right).
\]

The polynomial equation $h_{19} (t) = 0$ of 14$^{th}$ degree admits 14 solutions. Out of them 12 are convex (not written here) and two are real given by $ - 5.779189781$and $ - 0.1730346360$. These two real solutions are negative. Since we are working with $t > 0$, this means that there are no real positive solutions of the equation $h_{19} (t) = 0$. Thus we conclude that either $h_{19} (t) > 0$ or $h_{19} (t) < 0$, for all $t > 0$. In order to check it is sufficient to see for any particular value of $h_{19} (t)$, for example when $t = 1$. This gives $h_{19} (1) = \mbox{56448}$, hereby proving that $h_{19} (t) > 0$ for all $t > 0$, consequently, $w_{19} (x) > 0$, for all $x > 0$ giving $v_{19} (x) \ge 0$, for all $x > 0$. Finally, proving the required result.

\bigskip
\item \textbf{For }$\bf{\frac{P_5 + 6N_2 - G}{6} \le \frac{P_5 + 2N_2 }{3}}$\textbf{: }Equivalently,
we have to show that
\[
\textstyle{1 \over 6}\left( {P_5 - 2N_2 + G} \right) \ge 0.
\]

We can write $P_5 - 2N_2 + G = \textstyle{1 \over 2}b\,g_{20} \left( {a \mathord{\left/ {\vphantom {a b}} \right. \kern-\nulldelimiterspace} b} \right)$, where
\[
g_{20} (x) = \frac{u_{20} (x)}{\left( {\sqrt x + 1} \right)^2},
\]

\noindent with
\[
u_{20} (x) = 2\left( {x^2 + x^{3 / 2} + 4x + \sqrt x + 1} \right) - \sqrt
{2x + 2} \left( {\sqrt x + 1} \right)^3
\]

Now we shall show that $u_{20} (x) \ge 0$, $\forall x > 0$. Let us consider
\begin{align}
v_{20} (x) & = \left[ {2\left( {x^2 + x^{3 / 2} + 4x + \sqrt x + 1} \right)}
\right]^2 - \left[ {\sqrt {2x + 2} \left( {\sqrt x + 1} \right)^3}
\right]^2\notag\\
& = 2\left( {\sqrt x - 1} \right)^4\left( {x^2 + 2x^{3 / 2} + 4x + 2\sqrt x +
1} \right).\notag
\end{align}

Since $v_{20} (x) \ge 0$, giving $u_{20} (x) \ge 0$, $\forall x > 0$, hence proving the required result.

\bigskip
\item \textbf{For }$\bf{\frac{P_5 + 2N_2 }{3} \le \frac{6P_6 + 7H - 6N_2
}{7}}$\textbf{: }Equivalently, we have to show that
\[
\textstyle{1 \over {21}}\left( {18P_6 + 21H - 14N_2 - 18N_3 - 7P_5 } \right)
\ge 0.
\]

We can write $18P_6 + 21H - 14N_2 - 18N_3 - 7P_5 = \textstyle{1 \over 2}b\,g_{21} \left( {a \mathord{\left/ {\vphantom {a b}} \right. \kern-\nulldelimiterspace} b} \right)$, where
\[
g_{21} (x) = \frac{u_{21} (x)}{\left( {x + 1} \right)\left( {\sqrt x + 1}
\right)^2},
\]

\noindent with
\begin{align}
u_{21} (x) & = 2\left( {5x^3 + 18x^{5 / 2} + 9x^2 + 48x^{3 / 2} + 9x + 18\sqrt
x + 5} \right)\notag\\
& \hspace{20pt}  - 7\sqrt {2x + 2} \left( {\sqrt x + 1} \right)^3\left( {x + 1} \right).\notag
\end{align}

Now we shall show that $u_{21} (x) \ge 0$, $\forall x > 0$. Let us consider
\begin{align}
v_{21} (x) & = \left[ {2\left( {5x^3 + 18x^{5 / 2} + 9x^2 + 48x^{3 / 2} + 9x +
18\sqrt x + 5} \right)} \right]^2\notag\\
& \hspace{20pt}  - \left[ {7\sqrt {2x + 2} \left( {\sqrt x + 1} \right)^3\left( {x + 1}
\right)} \right]^2\notag\\
& = 2\left( {\sqrt x - 1} \right)^4\left( {\begin{array}{l}
 x^4 + 70x^{7 / 2} + 220x^3 + 210x^{5 / 2} + \\
 + 510x^2 + 210x^{3 / 2} + 220x + 70\sqrt x + 1 \\
 \end{array}} \right).\notag
\end{align}

Since $v_{21} (x) \ge 0$, giving $u_{21} (x) \ge 0$, $\forall x > 0$, hence proving the required result.

\bigskip
\item \textbf{For }$\bf{\frac{5H + 4S - 4P_4 }{5} \le \frac{12P_5 + 13P_2 - 12P_1
}{13}}$\textbf{: }We have to show that
\[
\textstyle{1 \over {65}}\left( {60P_5 + 65P_2 + 52P_4 - 60P_1 - 65H - 52S}
\right) \ge 0.
\]

We can write $60P_5 + 65P_2 + 52P_4 - 60P_1 - 65H - 52S = b\,g_{22} \left( {a \mathord{\left/ {\vphantom {a b}} \right. \kern-\nulldelimiterspace} b} \right)$, where
\[
g_{22} (x) = \frac{u_{22} (x)}{\left( {\sqrt x + 1} \right)^2\left( {x^3 +
1} \right)\left( {x^2 + 1} \right)},
\]

\noindent with
\begin{align}
u_{22} (x) & = \left( {\begin{array}{l}
 60x^7 + 203x^6 - 250x^{11 / 2} + 190x^5 + \\
 + 390x^{9 / 2} + 203x^4 - 760x^{7 / 2} + 203x^3 + \\
 + 390x^{5 / 2} + 190x^2 - 250x^{3 / 2} + 203x + 60 \\
 \end{array}} \right)\notag\\
& \hspace{20pt}  - 26\sqrt {2x^2 + 2} \left( {x^2 + 1} \right)\left( {x^3 + 1} \right)\left(
{\sqrt x + 1} \right)^2.\notag
\end{align}

Now we shall show that $u_{22} (x) \ge 0$, $\forall x > 0$. Let us consider
\begin{align}
v_{22} (x) & = \left( {\begin{array}{l}
 60x^7 + 203x^6 - 250x^{11 / 2} + 190x^5 + \\
 + 390x^{9 / 2} + 203x^4 - 760x^{7 / 2} + 203x^3 + \\
 + 390x^{5 / 2} + 190x^2 - 250x^{3 / 2} + 203x + 60 \\
 \end{array}} \right)^2\notag\\
& \hspace{20pt}  - \left[ {26\sqrt {2x^2 + 2} \left( {x^2 + 1} \right)\left( {x^3 + 1}
\right)\left( {\sqrt x + 1} \right)^2} \right]^2.\notag\\
& = \left( {\sqrt x - 1} \right)^2w_{22} (x),\notag
\end{align}

\noindent where
\[
w_{22} (x) = \left( {\begin{array}{l}
 2248x^{13} - 912x^{25 / 2} + 12176x^{12} - 10144x^{23 / 2} + \\
 + 26137x^{11} - 8506x^{21 / 2} + 93811x^{10} + 141228x^{19 / 2} + \\
 + 112187x^9 - 158954x^{17 / 2} + 249211x^8 + \\
 + 505084x^{15 / 2} + 137574x^7 - 471720x^{13 / 2} + \\
 + 137574x^6 + 505084x^{11 / 2} + 249211x^5 - 158954x^{9 / 2} + \\
 + 112187x^4 + 141228x^{7 / 2} + 93811x^3 - 8506x^{5 / 2} + \\
 + 26137x^2 - 10144x^{3 / 2} + 12176x - 912\sqrt x + 2248 \\
 \end{array}} \right).
\]

Now we shall show that $w_{22} (x) > 0$, $\forall x > 0$. Let us reorganize the positive and negative terms of $w_{22} (x)$ and apply again the argument given in Part 1, by considering
\begin{align}
w_{22a} (x) & = \left( {\begin{array}{l}
 2248x^{13} + 12176x^{12} + 26137x^{11} + 93811x^{10} + 141228x^{19 / 2} +
\\
 + 112187x^9 + 249211x^8 + 505084x^{15 / 2} + 137574x^7 + \\
 + 137574x^6 + 505084x^{11 / 2} + 249211x^5 + 112187x^4 + \\
 + 141228x^{7 / 2} + 93811x^3 + 26137x^2 + 12176x + 2248 \\
 \end{array}} \right)^2\notag\\
& \hspace{20pt}  - \left( {\begin{array}{l}
 912x^{25 / 2} + 10144x^{23 / 2} + 8506x^{21 / 2} + \\
 + 158954x^{17 / 2} + 471720x^{13 / 2} + 158954x^{9 / 2} + \\
 + 8506x^{5 / 2} + 10144x^{3 / 2} + 912\sqrt x \\
 \end{array}} \right)^2\notag\\
& = \left( {\begin{array}{l}
 5053504 + 53911552x + 939846800x^3 + 247264272x^2 + \\
 + 8394022562x^5 + 3439184256x^{9 / 2} + 9653410136x^{11 / 2} + \\
 + 54425545136x^7 + 3299451265x^4 + 634961088x^{7 / 2} + \\
 + 177455818032x^{35 / 2} + 388807633696x^{31 / 2} + 60361509952x^{37 / 2} +
\\
 + 54425545136x^{19} + 18127427183x^{20} + 60361509952x^{15 / 2} + \\
 + 61435331831x^8 + 38797285384x^{13 / 2} + 399552515016x^{11} + \\
 + 330074264536x^{23 / 2} + 179224041264x^{19 / 2} + 204575899694x^9 + \\
 + 388807633696x^{21 / 2} + 177455818032x^{17 / 2} + 169865767282x^{12} \\
 \end{array}} \right)\notag\\
& \hspace{20pt}  + \left( {\begin{array}{l}
 131742209367x^{10} + 483216926288x^{13} + 448903200968x^{25 / 2} + \\
 + 38797285384x^{39 / 2} + 9653410136x^{41 / 2} + 8394022562x^{21} + \\
 + 3299451265x^{22} + 939846800x^{23} + 247264272x^{24} + \\
 + 3439184256x^{43 / 2} + 634961088x^{5 / 2} + 53911552x^{25} + \\
 + 5053504x^{26} + 18127427183x^6 + 169865767282x^{14} + \\
 + 399552515016x^{15} + 131742209367x^{16} + \\
 + 330074264536x^{29 / 2} + 204575899694x^{17} + \\
 + 61435331831x^{18} + 448903200968x^{27 / 2} + 179224041264x^{33 / 2} \\
 \end{array}} \right).\notag
\end{align}

Thus the positivity of $w_{22a} (x)$ proves that $w_{22} (x) > 0$, $\forall x > 0$. This completes the proof of the result.

\bigskip
\item \textbf{For }$\bf{\frac{5H + 4S - 4P_4 }{5} \le \frac{4H + 5S - 5G}{4}}$\textbf{:
}We have to show that
\[
\textstyle{1 \over {20}}\left( {16P_4 + 9S - 25G} \right) \ge 0.
\]

We can write $16P_4 + 9S - 25G = \textstyle{1 \over 2}b\,g_{23} \left( {a \mathord{\left/ {\vphantom {a b}} \right. \kern-\nulldelimiterspace} b} \right)$, where
\[
g_{23} (x) = \frac{u_{23} (x)}{\left( {\sqrt x + 1} \right)^2},
\]

\noindent with
\[
u_{23} (x) = 9\left( {\sqrt x + 1} \right)^2\sqrt {2x^2 + 2}
 - \sqrt x \left[ {36\left( {x + 1} \right) + 14\left( {\sqrt x - 1}
\right)^2} \right].
\]

Now we shall show that $u_{23} (x) \ge 0$, $\forall x > 0$. Let us consider
\begin{align}
v_{23} (x) & = \left\{ {9\left( {\sqrt x + 1} \right)^2\sqrt {2x^2 + 2} }
\right\}^2\notag\\
& \hspace{20pt}  - \left\{ {\sqrt x \left[ {36\left( {x + 1} \right) + 14\left( {\sqrt x -
1} \right)^2} \right]} \right\}^2\notag\\
& = 2\left( {\sqrt x - 1} \right)^2\left( {\begin{array}{l}
 81x^3 + 486x^{5 / 2} + 127x^2 + \\
 + 1492x^{3 / 2} + 127x + 486\sqrt x + 81 \\
 \end{array}} \right).\notag
\end{align}

Since $v_{23} (x) \ge 0$, giving $u_{23} (x) \ge 0$, $\forall x > 0$, hence proving the required result.

\bigskip
\item \textbf{For }$\bf{\frac{7H + 6P_6 - 6N_3 }{7} \le \frac{12P_5 + 13P_2 - 12P_1
}{13}}$\textbf{: }We have to show that
\[
\textstyle{1 \over {91}}\left( {84P_5 + 91P_2 + 78N_3 - 84P_1 - 91H - 78P_6
} \right) \ge 0.
\]

We can write $84P_5 + 91P_2 + 78N_3 - 84P_1 - 91H - 78P_6 = b\,g_{24} \left( {a \mathord{\left/ {\vphantom {a b}} \right. \kern-\nulldelimiterspace} b} \right)$, where
\[
g_{24} (x) = \frac{u_{24} (x)}{\left( {\sqrt x + 1} \right)^2\left( {x^3 +
1} \right)\left( {x + 1} \right)},
\]

\noindent where
\[
u_{24} (x) = \left( {\sqrt x - 1} \right)^2\left( {\begin{array}{l}
 7\left( {x^5 + 2x^4 + 2x + 1} \right)\left( {\sqrt x - 1} \right)^2 + \\
 + 25x^6 + 16x^5 + 55x^4 + 322x^{7 / 2} + \\
 + 620x^3 + 322x^{5 / 2} + 55x^2 + 16x + 25 \\
 \end{array}} \right).
\]

Since $u_{24} (x) \ge 0,\,\forall x > 0$, this gives $u_{24} (x) \ge 0,\,\forall x > 0$, thereby proving the required result.

\bigskip
\item \textbf{For }$\bf{\frac{7H + 6P_6 - 6N_3 }{7} \le \frac{4H + 5S -
5G}{4}}$\textbf{: }We have to show that
\[
\textstyle{1 \over {28}}\left( {35S + 24N_3 - 35G - 24P_6 } \right) \ge 0.
\]

We can write $35S + 24N_3 - 35G - 24P_6 = \textstyle{1 \over 2}b\,g_{25} \left( {a \mathord{\left/ {\vphantom {a b}} \right. \kern-\nulldelimiterspace} b} \right)$, where
\[
g_{25} (x) = \frac{u_{25} (x)}{\left( {x + 1} \right)},
\]

\noindent with
\begin{align}
u_{25} (x) & = 35\left( {x + 1} \right)\sqrt {2x^2 + 2}\notag\\
& \hspace{20pt}  - \left[ {32x^2 + 38x^{3 / 2} + 16\sqrt x \left( {\sqrt x - 1} \right)^2 +
38\sqrt x + 32} \right].\notag
\end{align}

Now we shall show that $u_{25} (x) \ge 0$, $\forall x > 0$. Let us consider
\begin{align}
v_{25} (x) & = \left[ {35\left( {x + 1} \right)\sqrt {2x^2 + 2} } \right]^2\notag\\
& \hspace{20pt}  - \left[ {32x^2 + 38x^{3 / 2} + 16\sqrt x \left( {\sqrt x - 1} \right)^2 +
38\sqrt x + 32} \right]^2\notag\\
& = 2\left( {\sqrt x - 1} \right)^2\left( {\begin{array}{l}
 562x^3 + 151\left( {x^2 + 1} \right)\left( {\sqrt x - 1} \right)^2 + \\
 + 548x^2 + 1700x^{3 / 2} + 548x + 562 \\
 \end{array}} \right).\notag
\end{align}

Since $v_{25} (x) \ge 0$, giving $u_{25} (x) \ge 0$, $\forall x > 0$, hence proving the required result.

\bigskip
\item \textbf{For }$\bf{\frac{7H + 6P_6 - 6N_3 }{7} \le S}$\textbf{: }We have to show
that
\[
\textstyle{1 \over 7}\left( {7S + 6N_3 - 7H - 6P_6 } \right) \ge 0.
\]

We can write $7S + 6N_3 - 7H - 6P_6 = \textstyle{1 \over 2}b\,g_{26} \left({a \mathord{\left/ {\vphantom {a b}} \right. \kern-\nulldelimiterspace} b} \right)$, where
\[
g_{26} (x) = \frac{u_{26} (x)}{\left( {x + 1} \right)},
\]

\noindent with
\[
u_{26} (x) = 7\left( {x + 1} \right)\sqrt {2x^2 + 2}
 - \left( { - 8x^2 + 4x^{3 / 2} - 20x + 4\sqrt x - 8} \right).
\]

Now we shall show that $u_{26} (x) \ge 0$, $\forall x > 0$. Let us consider
\begin{align}
v_{26} (x) & = \left[ {7\left( {x + 1} \right)\sqrt {2x^2 + 2} } \right]^2
 - \left( { - 8x^2 + 4x^{3 / 2} - 20x + 4\sqrt x - 8} \right)^2\notag\\
& = 2\left( {\sqrt x - 1} \right)^2\left( {\begin{array}{l}
 17x^3 + 66x^{5 / 2} + 45x^2 + \\
 + 136x^{3 / 2} + 45x + 66\sqrt x + 17 \\
 \end{array}} \right).\notag
\end{align}

Since $v_{26} (x) \ge 0$, giving $u_{26} (x) \ge 0$, $\forall x > 0$, hence proving the required result.

\bigskip
\item \textbf{For }$\bf{\frac{P_6 + 2N_1 - G}{2} \le \frac{2P_4 + 5P_6 }{7}}$\textbf{:
}We have to show that
\[
\textstyle{1 \over {14}}\left( {4P_5 + 3P_6 + 7G - 14N_1 } \right) \ge 0.
\]

We can write $4P_5 + 3P_6 + 7G - 14N_1 = \textstyle{1 \over 2}b\,g_{27}
\left( {a \mathord{\left/ {\vphantom {a b}} \right.
\kern-\nulldelimiterspace} b} \right)$, where
\[
g_{27} (x) = \frac{\left( {\sqrt x - 1} \right)^2\left( {7x^2 + 12x^{3 / 2}
+ 26x + 12\sqrt x + 7} \right)}{\left( {x + 1} \right)\left( {\sqrt x + 1}
\right)^2},
\]

Since $g_{27} (x) \ge 0$, $\forall x > 0$, hence proving the required result.

\bigskip
\item \textbf{For }$\bf{\frac{P_6 + 2N_1 - G}{2} \le \frac{4N_2 + 5P_6 }{9}}$\textbf{:
}We have to show that
\[
\textstyle{1 \over {18}}\left( {18N_1 - P_6 - 9G - 8N_2 } \right) \ge 0.
\]

We can write $18N_1 - P_6 - 9G - 8N_2 = \textstyle{1 \over 2}b\,g_{28} \left( {a \mathord{\left/ {\vphantom {a b}} \right. \kern-\nulldelimiterspace} b} \right)$, where
\[
g_{28} (x) = \frac{u_{28} (x)}{\left( {x + 1} \right)},
\]

\noindent with
\[
u_{28} (x) = 7x^2 + 18x + 7 - 4\sqrt {2x + 2} \left( {\sqrt x + 1}
\right)\left( {x + 1} \right).
\]

Now we shall show that $u_{28} (x) \ge 0$, $\forall x > 0$. Let us consider
\begin{align}
v_{28} (x) & = \left( {7x^2 + 18x + 7} \right)^2 - \left[ {4\sqrt {2x + 2}
\left( {\sqrt x + 1} \right)\left( {x + 1} \right)} \right]^2\notag\\
& = \left( {\sqrt x - 1} \right)^4\left( {17x^2 + 4x^{3 / 2} + 38x + 4\sqrt x
+ 17} \right).\notag
\end{align}

Since $v_{28} (x) \ge 0$, giving $u_{28} (x) \ge 0$, $\forall x > 0$, hence proving the required result.

\bigskip
\item \textbf{For }$\bf{A + S - G \le \frac{12P_5 + 7P_2 - 12P_3 }{7}}$\textbf{: }We
have to show that
\[
\textstyle{1 \over 7}\left( {12P_5 + 7P_2 + 7G - 12P_3 - 7A - 7S} \right)
\ge 0.
\]

We can write $12P_5 + 7P_2 + 7G - 12P_3 - 7A - 7S = \textstyle{1 \over 2}b\,g_{29} \left( {a \mathord{\left/ {\vphantom {a b}} \right. \kern-\nulldelimiterspace} b} \right)$, where
\[
g_{29} (x) = \frac{u_{30} (x)}{\left( {\sqrt x + 1} \right)\left( {x^{3 / 2}
+ 1} \right)\left( {x^2 + 1} \right)},
\]

\noindent with
\begin{align}
u_{29} (x) & = \left( {\begin{array}{l}
 8x^5 + \left( {9x^4 + 48x^3 + 3x^2 + 48x + 9} \right)\left( {\sqrt x - 1}
\right)^2 \\
 x^{9 / 2} + 12x^4 + 35x^3 + + 35x^2 + 12x + \sqrt x + 8 \\
 \end{array}} \right)\notag\\
& \hspace{20pt}  - 7\sqrt {2x^2 + 2} \left( {x^{3 / 2} + 1} \right)\left( {x^2 + 1}
\right)\left( {\sqrt x + 1} \right).\notag
\end{align}

Now we shall show that $u_{29} (x) \ge 0$, $\forall x > 0$. Let us consider
\begin{align}
v_{29} (x) & = \left( {\begin{array}{l}
 8x^5 + \left( {9x^4 + 48x^3 + 3x^2 + 48x + 9} \right)\left( {\sqrt x - 1}
\right)^2 + \\
 + x^{9 / 2} + 12x^4 + 35x^3 + 35x^2 + 12x + \sqrt x + 8 \\
 \end{array}} \right)^2\notag\\
& \hspace{20pt}  - \left[ {7\sqrt {2x^2 + 2} \left( {x^{3 / 2} + 1} \right)\left( {x^2 + 1}
\right)\left( {\sqrt x + 1} \right)} \right]^2.\notag\\
& = \left( {\sqrt x - 1} \right)^4\left( {\begin{array}{l}
 186x^8 + 5\left( {x^7 + 1} \right)\left( {\sqrt x - 1} \right)^2 + 1346x^7
+ \\
 + 422x^{13 / 2} + 3614x^6 + 178x^{11 / 2} + 3185x^5 + \\
 + 1394x^{9 / 2} + 8022x^4 + 1394x^{7 / 2} + 3185x^3 + \\
 + 178x^{5 / 2} + 3614x^2 + 422x^{3 / 2} + 1346x + 186 \\
 \end{array}} \right).\notag
\end{align}

Since $v_{29} (x) \ge 0$, giving $u_{29} (x) \ge 0$, $\forall x > 0$, hence proving the required result.

\bigskip
\item \textbf{For }$\bf{\frac{6P_5 + 5G - 6H}{5} \le \frac{12P_5 + 7P_2 - 12P_3
}{7}}$\textbf{: }We have to show that
\[
\textstyle{1 \over {35}}\left( {18P_5 + 35P_2 + 42H - 60P_3 - 35G} \right)
\ge 0.
\]

We can write $18P_5 + 35P_2 + 42H - 60P_3 - 35G = b\,g_{30} \left( {a \mathord{\left/ {\vphantom {a b}} \right. \kern-\nulldelimiterspace} b} \right)$, where
\[
g_{30} (x) = \frac{u_{30} (x)}{\left( {\sqrt x + 1} \right)\left( {x^2 + 1}
\right)\left( {x^{3 / 2} + 1} \right)\left( {x + 1} \right)},
\]

\noindent where
\[
u_{30} (x) = \left( {\sqrt x - 1} \right)^4\left( {\begin{array}{l}
 18x^4 + 19x^{7 / 2} + 64x^3 + 124x^{5 / 2} + \\
 + 176x^2 + 124x^{3 / 2} + 64x + 19\sqrt x + 18 \\
 \end{array}} \right).
\]

Since $u_{30} (x) \ge 0,\,\forall x > 0$, this gives $u_{30} (x) \ge 0,\,\forall x > 0$, thereby proving the required result.

\bigskip
\item \textbf{For }$\bf{2A + G - 2P_4 \le \frac{12P_5 + 7P_2 - 12P_3 }{7}}$\textbf{:
}We have to show that
\[
\textstyle{1 \over 7}\left( {12P_5 + 7P_2 + 14P_4 - 12P_3 - 14A - 7G}
\right) \ge 0.
\]

We can write $12P_5 + 7P_2 + 14P_4 - 12P_3 - 14A - 7G = b\,g_{31} \left( {a
\mathord{\left/ {\vphantom {a b}} \right. \kern-\nulldelimiterspace} b}
\right)$, where
\[
g_{31} (x) = \frac{u_{32} (x)}{\left( {\sqrt x + 1} \right)\left( {x^2 + 1}
\right)\left( {x^{3 / 2} + 1} \right)},
\]

\noindent where
\[
u_{31} (x) = \left( {\sqrt x - 1} \right)^4\left( {\begin{array}{l}
 2x^3 + 3\left( {x^2 + 1} \right)\left( {\sqrt x - 1} \right)^2 + \\
 + 16x^2 + 21x^{3 / 2} + 16x + 2 \\
 \end{array}} \right).
\]

Since $u_{31} (x) \ge 0,\,\forall x > 0$, this gives $u_{31} (x) \ge 0,\,\forall x > 0$, thereby proving the required result.

\bigskip
\item \textbf{For }$\bf{A + S - G \le 6S + P_2 - 6A}$\textbf{: }We have to show that
$5S + P_2 + G - 7A \ge 0$. We can write $5S + P_2 + G - 7A = \textstyle{1
\over 2}b\,g_{32} \left( {a \mathord{\left/ {\vphantom {a b}} \right.
\kern-\nulldelimiterspace} b} \right)$, where
\[
g_{32} (x) = \frac{u_{32} (x)}{\left( {x^2 + 1} \right)},
\]

\noindent with
\[
u_{32} (x) = 5\sqrt {2x^2 + 2} \left( {x^2 + 1} \right)
 - \left( {7x^3 - 2x^{5 / 2} + 5x^2 - 2\sqrt x + 5x + 7} \right).
\]

Now we shall show that $u_{32} (x) \ge 0$, $\forall x > 0$. Let us consider
\begin{align}
v_{32} (x) & = \left[ {5\sqrt {2x^2 + 2} \left( {x^2 + 1} \right)} \right]^2
 - \left( {7x^3 - 2x^{5 / 2} + 5x^2 + 5x - 2\sqrt x + 7} \right)^2\notag\\
& = \left( {\sqrt x - 1} \right)^4\left( {\begin{array}{l}
 x^4 + 32x^{7 / 2} + 48x^3 + 5x^{3 / 2}\left( {\sqrt x - 1} \right)^2 + \\
 + 19x^{5 / 2} + 19x^{3 / 2} + 48x + 32\sqrt x + 1 \\
 \end{array}} \right).\notag
\end{align}

Since $v_{32} (x) \ge 0$, giving $u_{32} (x) \ge 0$, $\forall x > 0$, hence proving the required result.

\bigskip
\item \textbf{For }$\bf{\frac{6P_5 + 5G - 6H}{5} \le 6S + P_2 - 6A}$\textbf{: }We have
to show that
\[
\textstyle{1 \over 5}\left( {30S + 5P_2 + 6H - 30A - 6P_5 - 5G} \right) \ge
0.
\]

We can write $30S + 5P_2 + 6H - 30A - 6P_5 - 5G = b\,g_{33} \left( {a \mathord{\left/ {\vphantom {a b}} \right. \kern-\nulldelimiterspace} b} \right)$, where
\[
g_{33} (x) = \frac{u_{33} (x)}{\left( {\sqrt x + 1} \right)^2\left( {x^2 +
1} \right)\left( {x + 1} \right)},
\]

\noindent with
\begin{align}
u_{33} (x) & = 15\sqrt {2x^2 + 2} \left( {x^2 + 1} \right)\left( {\sqrt x + 1}
\right)^2\left( {x + 1} \right)\notag\\
& \hspace{20pt}  - \left( {\begin{array}{l}
 21x^5 + 35x^{9 / 2} + 56x^4 + 36x^{7 / 2} + 67x^3 + \\
 + 50x^{5 / 2} + 67x^2 + 36x^{3 / 2} + 56x + 35\sqrt x + 21 \\
 \end{array}} \right).\notag
\end{align}

Now we shall show that $u_{33} (x) \ge 0$, $\forall x > 0$. Let us consider
\begin{align}
v_{33} (x) &= \left[ {15\sqrt {2x^2 + 2} \left( {x^2 + 1} \right)\left(
{\sqrt x + 1} \right)^2\left( {x + 1} \right)} \right]^2\notag\\
& \hspace{20pt}  - \left( {\begin{array}{l}
 21x^5 + 35x^{9 / 2} + 56x^4 + 36x^{7 / 2} + 67x^3 + \\
 + 50x^{5 / 2} + 67x^2 + 36x^{3 / 2} + 56x + 35\sqrt x + 21 \\
 \end{array}} \right)^2\notag\\
& = \left( {\sqrt x - 1} \right)^4\left( {\begin{array}{l}
 9x^8 + 366x^{15 / 2} + 1433x^7 + 3540x^{13 / 2} + 6197x^6 + \\
 + 8892x^{11 / 2} + 10399x^5 + 10866x^{9 / 2} + 10676x^4 + \\
 + 10866x^{7 / 2} + 10399x^3 + 8892x^{5 / 2} + 6197x^2 + \\
 + 3540x^{3 / 2} + 1433x + 366\sqrt x + 9 \\
 \end{array}} \right).\notag
\end{align}

Since $v_{33} (x) \ge 0$, giving $u_{33} (x) \ge 0$, $\forall x > 0$, hence proving the required result.

\bigskip
\item \textbf{For }$\bf{2A + G - 2P_4 \le 6S + P_2 - 6A}$\textbf{: }We have to show that
\[
6S + P_2 + 2P_4 - G - 8A \ge 0.
\]

We can write $6S + P_2 + 2P_4 - G - 8A = b\,g_{34} \left( {a \mathord{\left/ {\vphantom {a b}} \right. \kern-\nulldelimiterspace} b} \right)$, where
\[
g_{34} (x) = \frac{u_{34} (x)}{\left( {x^2 + 1} \right)\left( {\sqrt x + 1}
\right)^2},
\]

\noindent with
\begin{align}
u_{34} (x) & = 3\sqrt {2x^2 + 2} \left( {x^2 + 1} \right)\left( {\sqrt x + 1}
\right)^2\notag\\
& \hspace{20pt}  - \left( {\begin{array}{l}
 4x^4 + 9x^{7 / 2} + x^3 + 7x^{5 / 2} + \\
 + 6x^2 + 7x^{3 / 2} + x + 9\sqrt x + 4 \\
 \end{array}} \right).\notag
\end{align}

Now we shall show that $u_{34} (x) \ge 0$, $\forall x > 0$. Let us consider
\begin{align}
v_{34} (x) & = \left[ {3\sqrt {2x^2 + 2} \left( {x^2 + 1} \right)\left( {\sqrt
x + 1} \right)^2} \right]^2\notag\\
& \hspace{20pt}  - \left( {\begin{array}{l}
 4x^4 + 9x^{7 / 2} + x^3 + 7x^{5 / 2} + \\
 + 6x^2 + 7x^{3 / 2} + x + 9\sqrt x + 4 \\
 \end{array}} \right)^2\notag\\
& = \left( {\sqrt x - 1} \right)^4\left( {\begin{array}{l}
 2x^6 + 8x^{11 / 2} + 39x^5 + 114x^{9 / 2} + \\
 + 149x^4 + 98x^{7 / 2} + 44x^3 + 98x^{5 / 2} + \\
 + 149x^2 + 114x^{3 / 2} + 39x + 8\sqrt x + 2 \\
 \end{array}} \right).\notag
\end{align}

Since $v_{34} (x) \ge 0$, giving $u_{34} (x) \ge 0$, $\forall x > 0$, hence proving the required result.

\bigskip
\item \textbf{For }$\bf{6S + P_2 - 6A \le \frac{16P_5 + 9P_1 - 16P_2 }{9}}$\textbf{:
}We have to show that
\[
\textstyle{1 \over 9}\left( {16P_5 + 9P_1 + 54A - 25P_2 - 54S} \right) \ge 0.
\]

We can write $16P_5 + 9P_1 + 54A - 25P_2 - 54S = b\,g_{35} \left( {a \mathord{\left/ {\vphantom {a b}} \right. \kern-\nulldelimiterspace} b} \right)$, where
\[
g_{35} (x) = \frac{u_{35} (x)}{\left( {\sqrt x + 1} \right)^2\left( {x^3 +
1} \right)\left( {x^2 + 1} \right)},
\]

\noindent with
\begin{align}
u_{35} (x) & = 27\sqrt {2x^2 + 2} \left( {x^2 + 1} \right)\left( {x^3 + 1}
\right)\left( {\sqrt x + 1} \right)^2\notag\\
& \hspace{20pt}  - \left( {\begin{array}{l}
 43x^7 + 54x^{13 / 2} + 70x^6 + 22x^{11 / 2} + 45x^5 + \\
 + 4x^{9 / 2} + 122x^4 + 144x^{7 / 2} + 122x^3 + 4x^{5 / 2} + \\
 + 45x^2 + 22x^{3 / 2} + 70x + 54\sqrt x + 43 \\
 \end{array}} \right).\notag
\end{align}

Now we shall show that $u_{35} (x) \ge 0$, $\forall x > 0$. Let us consider
\begin{align}
v_{35} (x) & = \left[ {27\sqrt {2x^2 + 2} \left( {x^2 + 1} \right)\left( {x^3
+ 1} \right)\left( {\sqrt x + 1} \right)^2} \right]^2\notag\\
& \hspace{20pt}  - \left( {\begin{array}{l}
 43x^7 + 54x^{13 / 2} + 70x^6 + 22x^{11 / 2} + 45x^5 + \\
 + 4x^{9 / 2} + 122x^4 + 144x^{7 / 2} + 122x^3 + 4x^{5 / 2} + \\
 + 45x^2 + 22x^{3 / 2} + 70x + 54\sqrt x + 43 \\
 \end{array}} \right)^2\notag\\
& = \left( {\sqrt x - 1} \right)^4\left( {\begin{array}{l}
 64x^{12} + 327\left( {x^{10} + 1} \right)\left( {x - 1} \right)^2 +
376x^{23 / 2} + 312x^{21 / 2} + \\
 + 11272x^{10} + 32320x^{19 / 2} + 50136x^9 + 51648x^{17 / 2} + \\
 + 42538x^8 + 38736x^{15 / 2} + 56978x^7 + 88128x^{13 / 2} + \\
 + 105160x^6 + 88128x^{11 / 2} + 56978x^5 + 38736x^{9 / 2} + \\
 + 42538x^4 + 51648x^{7 / 2} + 50136x^3 + 32320x^{5 / 2} + \\
 + 11272x^2 + 312x^{3 / 2} + 376\sqrt x + 64 \\
 \end{array}} \right).\notag
\end{align}

Since $v_{35} (x) \ge 0$, giving $u_{35} (x) \ge 0$, $\forall x > 0$, hence proving the required result.

\bigskip
\item \textbf{For }$\bf{\frac{12P_5 + 7P_2 - 12P_3 }{7} \le \frac{16P_5 + 9P_1 - 16P_2
}{9}}$\textbf{:} We have to show that
\[
\textstyle{1 \over {63}}\left( {4P_5 + 63P_1 + 108P_3 - 175P_2 } \right) \ge
0.
\]

We can write $4P_5 + 63P_1 + 108P_3 - 175P_2 = b\,g_{36} \left( {a \mathord{\left/ {\vphantom {a b}} \right. \kern-\nulldelimiterspace} b} \right)$, where
\[
g_{36} (x) = \frac{u_{36} (x)}{\left( {\sqrt x + 1} \right)\left( {x^2 + 1}
\right)\left( {x^{3 / 2} + 1} \right)\left( {x^3 + 1} \right)},
\]

\noindent where
\[
u_{36} (x) = \left( {\sqrt x - 1} \right)^4\left( {\begin{array}{l}
 4x^6 + 12x^{11 / 2} + 32x^5 + 168x^{9 / 2} + \\
 + 473x^4 + 705x^{7 / 2} + 764x^3 + 705x^{5 / 2} + \\
 + 473x^2 + 168x^{3 / 2} + 32x + 12\sqrt x + 4 \\
 \end{array}} \right).
\]

Since $u_{36} (x) \ge 0,\,\forall x > 0$, this gives $u_{36} (x) \ge 0,\,\forall x > 0$, thereby proving the required result.

\bigskip
\item \textbf{For }$\bf{\frac{P_6 + 3P_1 }{4} \le N_1 }$\textbf{:} We have to show that
\[
\textstyle{1 \over 4}\left( {4N_1 - P_6 - 3P_1 } \right) \ge 0.
\]

We can write $4N_1 - P_6 - 3P_1 = b\,g_{37} \left( {a \mathord{\left/
{\vphantom {a b}} \right. \kern-\nulldelimiterspace} b} \right)$, where
\[
g_{37} (x) = \frac{\sqrt x \left( {\sqrt x - 1} \right)^2\left( {2x + 3\sqrt
x + 2} \right)}{\left( {x - 1} \right)^2 + x}.
\]

Since $g_{37} (x) \ge 0, \, \forall x > 0$, hence proving the required result.

\bigskip
\item \textbf{For }$\bf{N_1 \le \frac{P_6 + 3G}{4}}$\textbf{:} We have to show that
\[
\textstyle{1 \over 4}\left( {P_6 + 3G - 4N_1 } \right) \ge 0.
\]

We can write $P_6 + 3G - 4N_1 = b\,g_{38} \left( {a \mathord{\left/{\vphantom {a b}} \right. \kern-\nulldelimiterspace} b} \right)$, where
\[
g_{38} (x) = \frac{\sqrt x \left( {\sqrt x - 1} \right)^2}{x + 1}.
\]

Since $g_{38} (x) \ge 0, \, \forall x > 0$, hence proving the required result.

\bigskip
\item \textbf{For }$\bf{\frac{P_6 + 3G}{4} \le \frac{S + 5N_1 }{6}}$\textbf{:}
Equivalently, we have to show that
\[
\textstyle{1 \over 6}\left( {2S + 10N_1 - 3P_6 - 9G} \right) \ge 0.
\]

We can write $2S + 10N_1 - 3P_6 - 9G = \textstyle{1 \over 2}b\,g_{39} \left( {a \mathord{\left/ {\vphantom {a b}} \right. \kern-\nulldelimiterspace} b} \right)$, where
\[
g_{39} (x) = \frac{u_{39} (x)}{x + 1},
\]

\noindent with
\begin{align}
u_{39} (x) & = 2\sqrt {2x^2 + 2} \left( {x + 1} \right)\notag\\
& \hspace{20pt}  - \left( {x^2 + 3x^{3 / 2} + 5\sqrt x \left( {\sqrt x - 1} \right)^2 +
3\sqrt x + 1} \right).\notag
\end{align}

Now we shall show that $u_{39} (x) \ge 0$, $\forall x > 0$. Let us consider
\begin{align}
v_{39} (x) & = \left[ {2\sqrt {2x^2 + 2} \left( {x + 1} \right)} \right]^2\notag\\
& \hspace{20pt}  - \left( {x^2 + 3x^{3 / 2} + 5\sqrt x \left( {\sqrt x - 1} \right)^2 +
3\sqrt x + 1} \right)^2\notag\\
& = \left( {\sqrt x - 1} \right)^4\left[ {7x^2 + 11\sqrt x \left( {\sqrt x -
1} \right)^2 + \sqrt x \left( {x + 1} \right) + 7} \right].\notag
\end{align}

Since $v_{39} (x) \ge 0$, giving $u_{39} (x) \ge 0$, $\forall x > 0$, hence proving the required result.

\bigskip
\item \textbf{For }$\bf{\frac{4P_5 + 9P_1 }{13} \le N_3 }$\textbf{:} We have to show that
\[
\textstyle{1 \over {13}}\left( {13N_3 - 4P_5 - 9P_1 } \right) \ge 0.
\]

We can write $13N_3 - 4P_5 - 9P_1 = \textstyle{1 \over 3}b\,g_{40} \left( {a
\mathord{\left/ {\vphantom {a b}} \right. \kern-\nulldelimiterspace} b}
\right)$, where
\[
g_{40} (x) = \frac{\left( {\sqrt x - 1} \right)^2\left( {\begin{array}{l}
 x^4 + 41x^{7 / 2} + 82x^3 + 108x^{5 / 2} + \\
 + 108x^2 + 108x^{3 / 2} + 82x + 41\sqrt x + 1 \\
 \end{array}} \right)}{\left( {\sqrt x + 1} \right)^2\left( {x^3 + 1}
\right)}.
\]

Since $g_{40} (x) \ge 0, \, \forall x > 0$, hence proving the required result.

\bigskip
\item \textbf{For }$\bf{\frac{P_5 + 5N_1 }{6} \le A}$\textbf{:} We have to show that
\[
\textstyle{1 \over 6}\left( {6A - P_5 - 5N_1 } \right) \ge 0.
\]
We can write $6A - P_5 - 5N_1 = \textstyle{1 \over 4}b\,g_{41} \left( {a
\mathord{\left/ {\vphantom {a b}} \right. \kern-\nulldelimiterspace} b}
\right)$, where
\[
g_{41} (x) = \frac{\left( {\sqrt x - 1} \right)^2\left( {\sqrt x + 3}
\right)\left( {3\sqrt x + 1} \right)}{\left( {\sqrt x + 1} \right)^2}.
\]

Since $g_{41} (x) \ge 0, \, \forall x > 0$, hence proving the required result.

\bigskip
\item \textbf{For }$\bf{\frac{2P_5 + 7P_2 }{9} \le N_1 }$\textbf{: }We have to show
that
\[
\textstyle{1 \over 9}\left( {9N_1 - 2P_5 - 7P_2 } \right) \ge 0.
\]

We can write $9N_1 - 2P_5 - 7P_2 = \textstyle{1 \over 4}b\,g_{42} \left( {a \mathord{\left/ {\vphantom {a b}} \right. \kern-\nulldelimiterspace} b} \right)$, where
\[
g_{42} (x) = \frac{\left( {\sqrt x - 1} \right)^2\left( {\begin{array}{l}
 x^3 + 38x^{5 / 2} + 85x^2 + \\
 + 112x^{3 / 2} + 85x + 38\sqrt x + 1 \\
 \end{array}} \right)}{\left( {\sqrt x + 1} \right)^2\left( {x^2 + 1}
\right)}.
\]

Since $g_{42} (x) \ge 0, \, \forall x > 0$, hence proving the required result.
\end{enumerate}
\end{proof}

\begin{remark} The results studied above in 1-42 parts can
equivalently be written as:
\begin{multicols}{2}
\begin{align}
& 1. \hspace{15pt} D_{P_6 S} \le \frac{1}{4}\left( {3D_{SP_4 } + 21D_{GP_4 } } \right).\notag\\
& 2. \hspace{15pt} D_{SN_2 } \le \frac{5}{3}D_{N_2 G}.\notag\\
& 3. \hspace{15pt} D_{P_6 N_2 } \le 3D_{N_2 G}.\notag\\
& 4. \hspace{15pt} D_{N_2 P_4 } \le \frac{1}{10}\left( {10D_{P_6 S} + D_{P_5 H} } \right).\notag\\
& 5. \hspace{15pt} D_{N_2 P_4 } \le \frac{1}{5}\left( {6D_{P_6 S} + D_{AN_2 } } \right).\notag\\
& 6. \hspace{15pt} D_{N_2 G} \le \frac{1}{20}\left( {10D_{P_5 N_2 } + 3D_{P_5 H} } \right).\notag\\
& 7. \hspace{15pt} D_{N_2 G} \le \frac{1}{4}\left( {2D_{P_5 N_2 } + D_{AP_4 } } \right).\notag\\
& 8. \hspace{15pt} D_{P_6 S} \le \frac{1}{2}D_{SP_4 }.\notag\\
& 9. \hspace{15pt} D_{P_6 N_1 } \le \frac{15}{14}D_{SP_4 }.\notag\\
& 10. \hspace{15pt} D_{P_6 S} \le \frac{1}{9}\left( {D_{SP_4 } + 8D_{N_1 P_4 } } \right).\notag\\
& 11. \hspace{15pt} D_{GP_4 } \le \frac{1}{2}D_{P_6 P_5}.\notag\\
& 12. \hspace{15pt} D_{AN_3 } \le \frac{1}{13}\left( {2D_{P_5 A} + 5D_{N_3 G} } \right).\notag\\
& 13. \hspace{15pt} D_{P_5 N_2 } \le 2D_{SN_3 } + 7D_{N_2 N_3 }.\notag\\
& 14. \hspace{15pt} D_{P_5 N_3 } \le \frac{5}{2}D_{SN_3 }.\notag
\end{align}

\begin{align}
& 15. \hspace{15pt} D_{SN_3 } \ge \frac{1}{20}\left( {7D_{N_3 H} + 3D_{P_5 H} } \right).\notag\\
& 16. \hspace{15pt} D_{AH} \le \frac{1}{45}\left( {36D_{P_6 N_3 } + 4D_{P_6 P_5 } } \right).\notag\\
& 17. \hspace{15pt} D_{P_6 A} \le \frac{9}{8}D_{AH}.\notag\\
& 18. \hspace{15pt} D_{P_6 H} \le \frac{1}{70}\left( {56D_{SP_4 } + 45D_{P_5 P_3 } } \right).\notag\\
& 19. \hspace{15pt} D_{P_6 N_2 } \le \frac{1}{42}\left( {27D_{P_5 P_3 } + 7D_{P_5 G} } \right).\notag\\
& 20. \hspace{15pt} D_{N_2 G} \le D_{P_5 N_2}.\notag\\
& 21. \hspace{15pt} D_{P_6 N_3 } \ge \frac{1}{18}\left( {14D_{N_2 H} + 7D_{P_5 H} } \right).\notag\\
& 22. \hspace{15pt} D_{P_5 P_1 } \ge \frac{1}{60}\left( {65D_{HP_2 } + 52D_{SP_4 } } \right).\notag\\
& 23. \hspace{15pt} D_{SP_4 } \le \frac{25}{16}D_{SG}.\notag\\
& 24. \hspace{15pt} D_{P_5 P_1 } \le \frac{1}{84}\left( {91D_{HP_2 } + 78D_{P_6 N_3 } } \right).\notag\\
& 25. \hspace{15pt} D_{P_6 N_3 } \le \frac{35}{24}D_{SG}.\notag\\
& 26. \hspace{15pt} D_{P_6 N_3 } \le \frac{7}{6}D_{SH}.\notag\\
& 27. \hspace{15pt} D_{N_1 G} \le \frac{1}{7}\left( {4D_{P_5 N_1 } + 3D_{P_6 N_1 } } \right).\notag\\
& 28. \hspace{15pt} D_{N_1 G} \ge \frac{1}{9}\left( {D_{P_6 N_1 } + 8D_{N_2 N_1 } } \right).\notag
\end{align}
\begin{align}
& 29. \hspace{15pt} D_{P_5 P_3 } \ge \frac{7}{12}\left( {D_{AP_2 } + D_{SG} } \right).\notag\\
& 30. \hspace{15pt} D_{GP_2 } \le \frac{1}{35}\left( {18D_{P_5 P_3 } + 42D_{HP_3 } } \right).\notag\\
& 31. \hspace{15pt} D_{P_5 P_3 } \ge \frac{1}{12}\left( {7D_{GP_2 } + 14D_{AP_4 } } \right).\notag\\
& 32. \hspace{15pt} D_{AG} + D_{A P_2}\le 5D_{SA}.\notag\\
& 33. \hspace{15pt} D_{SA} \ge \frac{1}{30}\left( {5D_{GP_2 } + 6D_{P_5 H} } \right).\notag\\
& 34 \hspace{15pt} D_{SA} \ge \frac{1}{6}\left( {2D_{AP_4 } + D_{GP_2 } } \right).\notag\\
& 35. \hspace{15pt} D_{P_5 P_2 } \ge \frac{1}{16}\left( {9D_{P_2 P_1 } + 54D_{SA} } \right).\notag
\end{align}
\begin{align}
& 36. \hspace{15pt} D_{P_2 P_1 } \le \frac{1}{63}\left( {4D_{P_5 P_2 } + 108D_{P_3 P_2 } } \right).\notag\\
& 37. \hspace{15pt} D_{P_6 N_1 } \le 3D_{N_1 P_1 }.\notag\\
& 38. \hspace{15pt} D_{N_1 G} \le \frac{1}{3}D_{P_6 N_1 }.\notag\\
& 39. \hspace{15pt} D_{P_6 N_1 } \le \frac{1}{3}\left( {2D_{SG} + 7D_{N_1 G} } \right).\notag\\
& 40. \hspace{15pt} D_{P_5 N_3 } \le \frac{9}{4}D_{N_3 P_1 }.\notag\\
& 41. \hspace{15pt} D_{P_5 A} \le 5D_{AN_1 }.\notag\\
& 42. \hspace{15pt} D_{P_5 N_1 } \le \frac{7}{2}D_{N_1 P_2 }.\notag
\end{align}
\end{multicols}
\end{remark}

\section{Connections with Divergence Measures}

Let
\[
\Gamma _n = \left\{ {P = (p_1 ,p_2 ,...,p_n )\left| {p_i > 0,\sum\limits_{i
= 1}^n {p_i = 1} } \right.} \right\},
\quad
n \ge 2,
\]

\noindent be the set of all complete finite discrete probability distributions. For
all $P,Q \in \Gamma _n $, the author \cite{tan2, tan3} proved the following inequalities:
\begin{align}
& \frac{1}{2}D_{AH} (P\vert \vert Q) \le I(P\vert \vert Q) \le 4D_{N_2 N_1 }
(P\vert \vert Q) \le \frac{4}{3}D_{N_2 G} (P\vert \vert Q)\notag\\
\label{eq16}
& \hspace{20pt}  \le D_{AG} (P\vert \vert Q) \le 4\,D_{AN_2 } (P\vert \vert Q) \le
\frac{1}{8}J(P\vert \vert Q) \le T(P\vert \vert Q),
\end{align}

\noindent where
\begin{align}
I(P\vert \vert Q) & = \frac{1}{2}\left[ {\sum\limits_{i = 1}^n {p_i \ln \left(
{\frac{2p_i }{p_i + q_i }} \right) + } \sum\limits_{i = 1}^n {q_i \ln \left(
{\frac{2q_i }{p_i + q_i }} \right)} } \right],\notag\\
J(P\vert \vert Q) & = \sum\limits_{i = 1}^n {(p_i - q_i )\ln \left( {\frac{p_i
}{q_i }} \right)}\notag
\intertext{and}
T(P\vert \vert Q)& = \sum\limits_{i = 1}^n {\left( {\frac{p_i + q_i }{2}}
\right)\ln \left( {\frac{p_i + q_i }{2\sqrt {p_i q_i } }} \right)} .\notag
\end{align}

The measures $I(P\vert \vert Q)$, $J(P\vert \vert Q)$ and $T(P\vert \vert
Q)$ are the respectively, the well-know \textit{Jensen-Shannon divergence},
$J - $\textit{divergence }and \textit{arithmetic and geometric mean divergence}. Moreover, $D_{AH}
(P\vert \vert Q) = \textstyle{1 \over 2}\Delta (P\vert \vert Q)$ and $D_{AG}
(P\vert \vert Q) = h(P\vert \vert Q)$, where $\Delta (P\vert \vert Q)$ and
$h(P\vert \vert Q)$ are the well-known \textit{triangular's} and \textit{Hellinger's} \textit{discriminations} respectively. These three measures satisfy the following equality:
\begin{equation}
\label{eq17}
J(P\vert \vert Q) = 4\left[ {I(P\vert \vert Q) + T(P\vert \vert Q)}
\right].
\end{equation}

Instead of $D_{( \cdot )} (a,b)$ as given in Section 1, here we are working with the probability distributions $P$ and $Q$, and writing $D_{( \cdot )} (a,b)$as $D_{( \cdot )} (P\vert \vert Q)$. Recently, author \cite{tan5} proved the following inequalities:
\begin{itemize}
\item [(i)] $\textstyle{2 \over 3}D_{AP_4 } \le I$;
\item [(ii)] $\textstyle{2 \over 3}D_{P_5 G} \le \textstyle{1 \over 8}J$;
\item [(iii)] $D_{P_5 A} \le \textstyle{1 \over 2}T.$
\end{itemize}

\begin{theorem} The following inequalities hold:
\begin{equation}
\label{eq18}
N_3 \le \frac{I + 4N_1 }{4} \le N_2 \le A \le \frac{2P_4 + 3I}{2} \le
\left\{ {\begin{array}{l}
 P_5 \le \frac{T + 2A}{2} \le \frac{3J + 16G}{16} \\
 S \\
 \end{array}} \right..
\end{equation}
\end{theorem}

\begin{proof} In view of (i)-(iii), we shall prove only the necessary parts.

\begin{enumerate}
\item \textbf{For }$\bf{N_3 \le \frac{I + 4N_1 }{4}}$\textbf{:} We can write it as
$4D_{N_3 N_1 } \le I$. Since $D_{N_3 N_1 } = \textstyle{1 \over 6}D_{AG} $,
we have to show that $\textstyle{2 \over 3}D_{AG} \le I$. According to
(\ref{eq8}), $D_{AG} \le D_{AP_4 } $. This together with (i) proves the requires
result.

\item \textbf{For }$\bf{\frac{2P_4 + 3I}{2} \le P_5 }$\textbf{:} We can write it as $I
\le \textstyle{2 \over 3}D_{P_5 P_4 } $. Since, $D_{P_5 P_4 } = 2D_{AG} $
and $I \le 4D_{N_2 N_1 } \le D_{AG} $. This implies that $I \le D_{AG} =
\textstyle{1 \over 2}D_{P_5 P_4 } \le \textstyle{2 \over 3}D_{P_5 P_4 } $.
This gives the requires result.

\item \textbf{For }$\bf{\frac{T + 2A}{2} \le \frac{3J + 16G}{16}}$\textbf{:} We have to
show that $\textstyle{3 \over {16}}J + G - \textstyle{1 \over 2}T - A \ge
0$, i.e., $A - G \le \textstyle{3 \over {16}}J - \textstyle{1 \over 2}T$.
Since $J = 4\left( {I + T} \right)$, $D_{AG} \le I$ and $\textstyle{1 \over
8}J \le T$, combining we get the required result.

\item \textbf{For }$\bf{\frac{2P_4 + 3I}{2} \le S}$\textbf{: }We can write $I \le
\textstyle{2 \over 3}D_{SP_4 } $. We know that $I \le 4D_{N_2 N_1 } $. We
shall prove that $I \le 4D_{N_2 N_1 } \le \textstyle{2 \over 3}D_{SP_4 } $.
In order to show this we need to show that $4D_{N_2 N_1 } \le \textstyle{2
\over 3}D_{SP_4 } $, i.e.,
\[
\textstyle{1 \over 3}\left( {2S + 12N_1 - 2P_4 - 12N_2 } \right) \ge 0.
\]

We can write $2S + 12N_1 - 2P_4 - 12N_2 = \sum\limits_{i = 1}^n {q_i }
\,g_{43} \left( {{p_i } \mathord{\left/ {\vphantom {{p_i } {q_i }}} \right.
\kern-\nulldelimiterspace} {q_i }} \right)$, where
\[
g_{43} (x) = \frac{u_{43} (x)}{\left( {\sqrt x + 1} \right)^2},
\]

\noindent with
\begin{align}
u_{43} (x) & = 3x^2 + 12x^{(3 / 2)} + 10x + 12\sqrt x + 3 +\notag\\
& \hspace{20pt} + \sqrt {2x^2 + 2} \left( {\sqrt x + 1} \right)^2 - 3\sqrt {2x + 2} \left(
{\sqrt x + 1} \right)^3.\notag
\end{align}

Now we shall show that $u_{43} (x) \ge 0$, $\forall x > 0$. We shall apply
twice the argument given in Part 1 of Theorem 2.1. Let us consider
\begin{align}
v_{43} (x) & = \left( {\begin{array}{l}
 3x^2 + 12x^{(3 / 2)} + 10x + 12\sqrt x + \\
 + 3 + \sqrt {2x^2 + 2} \left( {\sqrt x + 1} \right)^2 \\
 \end{array}} \right)^2\notag\\
&\hspace{20pt} - \left[ {3\sqrt {2x + 2} \left( {\sqrt x + 1} \right)^3} \right]^2\notag\\
& = 2\sqrt {2x^2 + 2} \left( {\begin{array}{l}
 3x^3 + 18x^{5 / 2} + 37x^2 + \\
 + 44x^{3 / 2} + 37x + 18\sqrt x + 3 \\
 \end{array}} \right)\notag\\
& \hspace{20pt} - \left( {\begin{array}{l}
 7x^4 + 28x^{7 / 2} + 72x^3 + 148x^{5 / 2} + \\
 + 130x^2 + 148x^{3 / 2} + 72x + 28\sqrt x + 7 \\
 \end{array}} \right).\notag
\end{align}

Let us consider again
\begin{align}
v_{43a} (x) & = \left[ {2\sqrt {2x^2 + 2} \left( {\begin{array}{l}
 3x^3 + 18x^{5 / 2} + 37x^2 + \\
 + 44x^{3 / 2} + 37x + 18\sqrt x + 3 \\
 \end{array}} \right)} \right]^2\notag\\
& \hspace{20pt} - \left( {\begin{array}{l}
 7x^4 + 28x^{7 / 2} + 72x^3 + 148x^{5 / 2} + \\
 + 130x^2 + 148x^{3 / 2} + 72x + 28\sqrt x + 7 \\
 \end{array}} \right)^2\notag\\
& = \left( {\sqrt x - 1} \right)^2\left( {\begin{array}{l}
 + 23x^7 + 518x^{13 / 2} + 3589x^6 + 13324x^{11 / 2} + \\
 + 33239x^5 + 60922x^{9 / 2} + 85773x^4 + \\
 + 96744x^{7 / 2} + 85773x^3 + 60922x^{5 / 2} + \\
 + 33239x^2 + 13324x^{3 / 2} + 3589x + 518\sqrt x + 23 \\
 \end{array}} \right).\notag
\end{align}

Thus the non-negativity of $v_{43a} (x)$ proves that $v_{43} (x) \ge 0$, $\forall
x > 0$, thereby proving that $u_{43} (x) \ge 0$, $\forall x > 0$. This
completes the proof of the result.
\end{enumerate}
\end{proof}

\begin{remark} As a consequence of above results we have the following
new inequalities:
\begin{itemize}
\item[(i)] $4D_{N_3 N_1 } \le I \le \frac{2}{3}\left\{ {\begin{array}{l}
 D_{P_5 P_4 } \\
 D_{SP_4 } \\
 \end{array}} \right.$;
\item[(ii)] $\frac{2}{5}D_{SP_4 } \le I \le \frac{2}{3}D_{SP_4 }$;
\item[(iii)] $\frac{2}{3}h \le I \le h$;
\item[(iv)] $I \le \frac{1}{8}J \le T \le \frac{1}{4}J$.
\end{itemize}
\end{remark}

\end{document}